\documentclass[a4paper]{article} 
\usepackage{amsmath, amssymb, amsthm}
\usepackage[mathscr]{euscript}     
\usepackage{dsfont} 
\usepackage[usenames,dvipsnames]{color}
\usepackage{graphicx} 
\usepackage{qip_compatible} 
\usepackage[section]{thmenvironments} 
\usepackage{layoutcommands} 
\usepackage{underscore} 
\usepackage[shortcuts]{extdash} 
\usepackage{authblk} 
\usepackage{appendix} 
\usepackage{microtype} 
\usepackage{enumitem} 
\setenumerate[1]{label=\textbf{\arabic*.}, ref=\arabic*}
\setenumerate[2]{label=\theenumi.\arabic*., ref=\theenumi.\arabic*}
\setenumerate[3]{label=\theenumi.\theenumii.\arabic*., ref=\theenumi.\theenumii.\arabic*}
\usepackage[margin=10pt,font=small,labelfont=bf,labelsep=endash,hypcap=true]{caption}
\usepackage[margin=10pt,font=small,labelfont=bf,labelformat=parens,labelsep=space,hypcap=false]{subcaption}
\usepackage{tikz} 
\usetikzlibrary{fit} 

\tikzset{sArrow/.style={->,>=stealth,thick}}
\tikzset{gArrow/.style={->,>=stealth,thick,gray}}
\tikzset{arrowLabel/.style={auto}}
\tikzset{blocResource/.style={draw,minimum width=4cm,minimum height=1.9cm}}
\tikzset{brnode/.style={minimum width=3.4cm,minimum height=.2cm}}
\tikzset{largeResource/.style={draw,minimum width=3.736cm,minimum height=3.25cm}}
\tikzset{medResource/.style={draw,minimum width=3.736cm,minimum height=1.75cm}}
\tikzset{thinResource/.style={draw,minimum width=1.618*2cm,minimum height=1cm}}
\tikzset{protocol/.style={draw,minimum width=1.545cm,minimum height=2.5cm}}
\tikzset{protocolLong/.style={draw,minimum height=1cm,minimum width=2.8cm}}
\tikzset{pnode/.style={minimum width=1cm,minimum height=.5cm}}
\tikzset{simulator/.style={draw,minimum width=2.8cm,minimum height=1.7cm}}

\usepackage[linktoc=all]{hyperref} 
\usepackage[open,numbered]{bookmark} 
\usepackage{breakurl} 
\usepackage{cite} 
\usepackage{hyperlinks} 
\usepackage{eprint} 

\usepackage{algorithm}
\floatname{algorithm}{Protocol}
\usepackage{complexity}
\input{Qcircuit.tex}
\makeatletter
\newcommand{\vast}{\bBigg@{3}}
\newcommand{\Vast}{\bBigg@{6.5}}
\makeatother
\newcommand{\djj}{d\kern-0.4em\char"16\kern-0.1em}

\hypersetup{
  pdftitle={Composable security of delegated quantum computation},
  pdfauthor={Vedran Dunjko, Joseph Fitzsimons, Christopher Portmann, Renato Renner},
  pdfstartview={FitH}
}

\newcommand{\err}{\text{err}}
\newcommand{\ok}{\text{ok}}
\newcommand{\inn}{\text{in}}
\newcommand{\out}{\text{out}}
\newcommand{\blind}{\text{blind}}
\newcommand{\verif}{\text{verif}}
\newcommand{\ctrlZ}{\text{ctrl-}Z}
\newcommand{\aR}{\mathscr{R}}
\newcommand{\aS}{\mathscr{S}}
\newcommand{\aT}{\mathscr{T}}
\newcommand{\aD}{\mathscr{D}}

\title{Composable security of delegated quantum computation}

\author[1,2]{Vedran Dunjko\thanks{Now at: Institute for Theoretical
    Physics, University of Innsbruck, Technikerstra{\ss}e 25, A-6020
    Innsbruck, Austria. \mailto{vedran.dunjko@uibk.ac.at}}}
\author[3,4]{Joseph F.\ Fitzsimons\email{joe.fitzsimons@nus.edu.sg}}
\author[5,6]{Christopher Portmann\email{chportma@phys.ethz.ch}}
\author[5]{Renato Renner\email{renner@phys.ethz.ch}}

\affil[1]{School of Informatics, University of Edinburgh, Edinburgh EH8 9AB, U.K.}
\affil[2]{Division of Molecular Biology, Ru\djj er Bo\v{s}kovi\'{c}
  Institute, Bijeni\v{c}ka cesta 54, P.P. 180, 10002 Zagreb, Croatia.}
\affil[3]{Singapore University of Technology and Design, 20 Dover Drive, Singapore 138682.}
\affil[4]{Centre for Quantum Technologies, National University of Singapore, Block S15, 3 Science Drive 2, Singapore 117543.}
\affil[5]{Institute for Theoretical Physics, ETH Zurich, 8093 Zurich, Switzerland.}
\affil[6]{Group of Applied Physics, University of Geneva, 1211 Geneva, Switzerland.}

\date{\today}

\begin{document}

\pdfbookmark[1]{Title page}{titlepage}

\maketitle

\begin{abstract}
  Delegating difficult computations to remote large computation
  facilities, with appropriate security guarantees, is a possible
  solution for the ever\-/growing needs of personal computing
  power. For delegated computation protocols to be usable in a larger
  context \--- or simply to securely run two protocols in parallel
  \--- the security definitions need to be composable. Here, we define
  composable security for delegated quantum computation.  We
  distinguish between protocols which provide only \emph{blindness}
  \--- the computation is hidden from the server \--- and those that
  are also \emph{verifiable} \--- the client can check that it has
  received the correct result. We show that the composable security
  definition capturing both these notions can be reduced to a
  combination of several distinct ``trace\-/distance\-/type'' criteria
  \--- which are, individually, non\-/composable security definitions.

  Additionally, we study the security of some known delegated quantum
  computation protocols, including Broadbent, Fitzsimons and Kashefi's
  Universal Blind Quantum Computation protocol. Even though these
  protocols were originally proposed with insufficient security
  criteria, they turn out to still be secure given the stronger
  composable definitions.
\end{abstract}


\newpage
\phantomsection
\pdfbookmark[1]{\contentsname}{sec:toc}
\tableofcontents
\newpage

\section{Introduction}
\label{sec:intro}

\subsection{Background}
\label{sec:intro.background}

It is unknown in what form quantum computers will be built. One
possibility is that large quantum servers may take a role similar to
that occupied by massive superclusters today. They would be available
as important components in large information processing clouds,
remotely accessed by clients using their home-based simple
devices. The issue of the security and the privacy of the computation
is paramount in such a setting.

Childs~\cite{Chi05} proposed the first such delegated quantum
computation (DQC) protocol, which hides the computation from the
server, i.e., the computation is blind. This was followed by Arrighi
and Salvail~\cite{AS06}, who introduced a notion of verifiability \---
checking that the server does what is expected \--- but only for a
restricted class of public functions. In recent years, this problem
has gained a lot of interest, with many papers proposing new
protocols, e.g.,
\cite{BFK09,ABE10,MDK10,DKL12,MF12,MF13,FK12,Mor12,SKM13,MK13,GMMR13,CMK13,MPDF13,Mor14},
and even small\-/scale experimental realizations~\cite{BKBFZW12,BFKW13}.

However, with the exception of recent work by Broadbent, Gutoski and
Stebila~\cite{BGS13}, none of the previous DQC papers consider the
\emph{composability} of the protocol. They prove security by showing
that the states held by the client and server fulfill some local
condition: the server's state must not contain any information about
the input and the client's final state must either be the correct
outcome or an error flag. Even though this means that the server
cannot \--- from the information leaked during a single execution of
the protocol in an isolated environment \--- learn the computation or
produce a wrong output without being detected, it does not guarantee
any kind of security in any realistic setting. In particular, if a
server treats two requests simultaneously or if the delegated
computation is used as part of a larger protocol (such as the quantum
coins of Mosca and Stebila~\cite{MS10}), these works on DQC cannot be
used to infer security. A \emph{composable security} framework must be
used for a protocol to be secure in an arbitrary environment. In the
following, we use the expression \emph{local} to denote the
non\-/composable security conditions previously used for DQC. This
term is chosen, because these criteria consider the state of a (local)
subsystem, instead of the global system as seen by a distinguisher in
composable security.\footnote{Standard terms for various forms of
  non\-/composable security, e.g., \emph{stand\-/alone} or
  \emph{sequential}, have precise definitions which do not apply to
  these security criteria.}

In fact, exactly these local properties have been proven to be
insufficient to define secure communication. There exist protocols
which are shown to both encrypt and authenticate messages by
fulfilling local criteria equivalent to the ones used in DQC \--- the
scheme is secure if the eavesdropper obtains no information about the
message from the ciphertext and authentic if the receiver either gets
the original message or an error flag. But if the eavesdropper learns
whether the message was transmitted faithfully or not, she learns some
information about this message~\cite{BN00,Kra01,MT10}. Since any
secure communication protocol can be seen as delegated computation for
the identity operation \--- Eve is required to apply the identity
operation to the message, but may cheat and try to learn or modify it
\--- there is a strict gap between security of DQC and previously used
local criteria.\footnoteremember{fn:gap}{An alternative example of
  this gap is as follows. The task is to compute a witness for a
  positive instance of an $\NP$ problem, and we do so with the
  following protocol: the server simply picks a witness at random and
  sends it to the client. Although the protocol does not achieve
  completeness, it appears to be sound: the protocol obviously does
  not leak any information about the input, since no information is
  sent from the client to the server. The client can also verify that
  the solution received is correct, and never accepts a wrong
  answer. But if the server ever learns whether the witness was
  accepted \--- e.g., it is composed with another protocol which makes
  this information public \--- he learns something about the input. If
  there are only two choices for the input with distinct witnesses, he
  learns exactly which one was used.}

Composable frameworks have the further advantage that they require the
interaction between different entities to be modeled explicitly, and
often make hidden assumptions apparent. For example, it came as a
surprise when Barrett et al.~\cite{BCK13} showed that device
independent quantum key distribution (DIQKD) is insecure if untrusted
devices (with internal memory) are used more than once. It is however
immediate when one models the security of DIQKD in a composable
framework, that existing security proofs make the assumption that
devices are used only once. Another example, the security definitions
of zero\-/knowledge protocols~\cite{Gol01} and coin
expansion~\cite{HMU06} make the assumption that the dishonest party
executes his protocol without interaction with the
environment.\footnote{The security definitions for these two problems
  are instances of what is generally known as \emph{stand\-/alone
    security}~\cite{Gol04}.} By explicitly modeling this
restriction,\footnote{This can be done by introducing a resource \---
  e.g., a trusted third party \--- that runs whatever circuits Alice
  and Bob give it in an isolated system, then returns the transcript
  of the protocol to both players.} these proofs can be lifted to a
composable framework. This has been used by, e.g., Unruh~\cite{Unr11},
who explicitly limits the number of parallel executions of a protocol
to achieve security in the bounded storage model.

Correctly defining the security of a cryptographic task is fundamental
for a protocol and proof to have any usefulness or even meaning. In
this paper we solve this problem for DQC, which has been open since the
first version of Childs's work~\cite{Chi05} was made available in 2001.



\subsection{Scope and security of DQC}
\label{sec:intro.scope}

A common feature of all DQC protocols is that the client, while not
being capable of full\-/blown quantum computation, has access to
limited quantum\-/enriched technology, which she needs to interact
with the server.  One of the key points upon which the different DQC
protocols vary, is the complexity and the technical feasibility of the
aforementioned quantum\-/enriched technology.  In particular, in the
proposal of Childs~\cite{Chi05}, the client has quantum memory, and
the capacity to perform local Pauli operations. The protocol of
Arrighi and Salvail~\cite{AS06} requires the client to have the
ability to generate relatively involved superpositions of multi-qubit
states, and perform a family of multi-qubit measurements.  Aharonov,
Ben-Or and Eban~\cite{ABE10}, for the purposes of studying quantum
prover interactive proof systems, considered a DQC protocol in which
the client has a constant\-/sized quantum computer. The blind DQC
protocol proposed by Broadbent, Fitzsimons and Kashefi\cite{BFK09} has
arguably the lowest requirements on the client. In particular, she
does not need any quantum memory,\footnote{This holds in the case of
  classical input and output. If quantum inputs and/or outputs are
  considered, then the client has to be able to apply a quantum
  one-time pad to the input state, and also decrypt a quantum one-time
  pad of the output state.} and is only required to prepare single
qubits in separable states randomly chosen from a small finite set
analogous to the BB84 states.\footnote{The states needed by the
  protocol of \cite{BFK09} are
  $\{(\sket{0}+e^{ik\pi/4}\sket{1})/\sqrt{2}\}_k$ for $k \in
  \{0,\dotsc,7\}$.} Alternatively, Morimae and
Fujii~\cite{MF13,Mor14} propose a DQC protocol in which the client
only needs to measure the qubits she receives from the server to
perform the computation.

A second important distinction between these protocols is in the types
of problems the protocol empowers the client to solve. Most protocols,
e.g., \cite{Chi05,ABE10,BFK09,FK12,MF13,Mor14}, allow a client to
perform universal quantum computation, whereas in \cite{AS06} the
client is restricted to the evaluation of
random\-/verifiable\footnote{Roughly speaking, a function $f$ is
  random\-/verifiable if pairs of instances and solutions $(x,f(x))$
  can be generated efficiently, where $x$ is sampled according to the
  uniform distribution from the function's domain.} functions.

Finally, an important characteristic of these protocols is the flavor
of security guaranteed to the client. Here, one is predominantly
interested in two distinct features: privacy of computation (generally
referred to as blindness) and verifiability of computation. Blindness
characterizes the degree to which the computational input and output,
and the computation itself, remain hidden from the server. This is the
main security concern of, e.g., \cite{Chi05,BFK09,MF13}. Verifiability
ensures that the client has means of confirming that the final output
of the computation is correct. In addition to blindness, some form of
verifiability is given by, e.g., \cite{AS06,ABE10,FK12,Mor14}. These
works do however not concern themselves with the cryptographic
soundness of their security notions. In particular, none of them
consider the issue of composability of DQC. A notable exception is the
recent work of Broadbent, Gutoski and Stebila~\cite{BGS13}, who,
independently from our work, prove that a variant of the DQC protocol
of Aharonov, Ben-Or and Eban~\cite{ABE10} provides composable
security.\footnote{The work of Broadbent et al.~\cite{BGS13} is on
  one-time programs. Their result on the composability of DQC is
  obtained by modifying their main one-time program protocol and
  security proof so that it corresponds to a variant of the DQC
  protocol from \cite{ABE10}.}

\subsection{Composable security}
\label{sec:intro.comp}

The first frameworks for defining composable security were proposed
independently by Canetti~\cite{Can01,Can13} and by Backes, Pfitzmann
and Waidner~\cite{PW01,BPW04,BPW07}, who dubbed them \emph{Universally
  Composable (UC) security} and \emph{Reactive Simulatability},
respectively. These security notions have been extended to the quantum
setting by Ben-Or and Mayers~\cite{BM04} and
Unruh~\cite{Unr04,Unr10}.

More recently, Maurer and Renner proposed a new composable framework,
Abstract Cryptography (AC)~\cite{MR11}. Unlike its predecessors
that use a bottom\-/up approach to defining models of computation,
algorithms, complexity, efficiency, and then security of cryptographic
schemes, the AC approach is top\-/down and axiomatic, where lower
abstraction levels inherit the definitions and theorems (e.g., a
composition theorem) from the higher level, but the definition or
concretization of low levels is not required for proving theorems at
the higher levels. In particular, it is not hard\-/coded in the
security notions of AC whether the underlying computation model is
classical or quantum, and this framework can be used equally for both.

Even though these frameworks differ considerably in their approach,
they all share the common notion that composable security is defined
by the distance between the real world setting and an ideal setting
in which the cryptographic task is accomplished in some perfect
way. We use AC in this work, because it simplifies the security
definitions by removing many notions which are not necessary at that
level of abstraction. But the same results could have been proven
using another framework, e.g., a quantum version of UC
security~\cite{Unr10}.

\subsection{Results}
\label{sec:intro.new}

In this paper, we define a composable framework for analyzing the
security of delegated quantum computing, using the aforementioned AC
framework~\cite{MR11}. We model DQC in a generic way, which is
independent of the computing requirements or universality of the
protocol, and encompasses to the best of our knowledge all previous
work on DQC. We then define composable blindness and composable
verifiability in this framework. The security definitions are thus
applicable to any DQC protocol fitting in our model.

We study the relations between local security criteria used in
previous works \cite{Chi05,AS06,ABE10,BFK09,MF13,FK12,Mor14} and
composable security of DQC. We show that by strengthening the existing
notion of local\-/verifiability, we can close the gap between these
local criteria and composable security of DQC. To do this we introduce
the notion of \emph{independent} local\-/verifiability. Intuitively,
this captures the idea that the acceptance probability of the client
should not depend on the input or computation performed, but rather
only on the activities of the (dishonest) server.  Our main
theorem is as follows.

\begin{thm}
\label{thm:main}
If a DQC protocol implementing a unitary\footnote{Any quantum
  operation can be written as a unitary on a larger system,
  effectively allowing this theorem to apply to all quantum
  operations, see \remref{rem:sa.comp}.} transformation provides
$\eps_{\text{bl}}$\=/local\-/blindness and
$\eps_{\text{ind}}$\=/independent
$\eps_{\text{ver}}$\=/local\-/verifiability for all inputs
$\psi_{A_CA_Q}$, where $A_C$ is classical and $A_Q$ is quantum, then
it is $\delta N^2$\=/secure, where $\delta =
4\sqrt{2\eps_{\text{ver}}} + 2\eps_{\text{bl}} + 2\eps_{\text{ind}}$
and $N = \dim \hilbert_{A_Q}$.
\end{thm}

Note that by choosing the parameters such that $\delta$ is
exponentially small in the size of the quantum input ($\log N$)
negates the factor $N^2$ blow-up in the overall error (see also
\remref{rem:failure}).

Proving that a DQC protocol is secure then reduces to proving that
these local criteria are satisfied.\footnote{This is similar in
  nature to the result on the composable security of quantum key
  distribution (QKD)~\cite{PR14}, 
  which shows that a QKD protocol that satisfies definitions of
  \emph{robustness}, \emph{correctness} and \emph{secrecy} is secure
  in a composable sense. These individual notions are all expressed
  with trace\-/distance\-/type criteria, e.g., a QKD protocol is
  $\eps$\=/secret if $(1-p_{\text{abort}})\trnorm{\rho_{KE} - \tau_{K}
    \tensor \rho_E} \leq \eps,$ where $p_{\text{abort}}$ is the
  probability of aborting, $\rho_{KE}$ the joint state of the final
  key and the eavesdropper's system and $\tau_K$ is the fully mixed
  state. To prove that a QKD protocol is secure, it is thus sufficient
  to prove that it satisfies these individual notions.} For instance,
the protocols of Fitzsimons and Kashefi~\cite{FK12} and
Morimae~\cite{Mor14} are shown to satisfy definitions of
local\-/correctness, local\-/blindness and local\-/verifiability,
equivalent to the ones considered here. To prove that these protocols
are secure, it only remains to show that they also satisfy the
stronger notion of \emph{independent} local\-/verifiability introduced
in this work, which we sketch in \appendixref{app:app}.

Finally, we analyze the security of two protocols \--- Broadbent,
Fitzsimons and Kashefi \cite{BFK09} and Morimae and Fujii~\cite{MF13}
\--- that do not provide any form of verifiability, so the generic
reduction cannot be used. Instead we directly prove that both these
protocols satisfy the definition of composable blindness, without
verifiability (in Theorems~\ref{thm:bfk} and \ref{thm:oneway} on
pages~\pageref{thm:bfk} and \pageref{thm:oneway}).

Interestingly \--- and somewhat unexpectedly \--- even though the
local security definitions used in previous works are insufficient to
guarantee composable security, the previously proposed protocols
studied in this work are all still secure given the stronger security
notions.

\subsection{Other related work}
\label{sec:intro.related}

The blind DQC protocol of \cite{BFK09} has been getting considerable
attention in both the experimental and theoretical scientific
community. Due to the relatively modest requirements on the client, a
small\-/scale experimental realization of this protocol has already
been demonstrated \cite{BKBFZW12}. And even more recently, an
experimental demonstration of the protocol of \cite{FK12} \--- which
includes verifiability \--- has been performed as well \cite{BFKW13}.

Various theoretical modifications of this protocol have been
proposed. For instance, the settings where the client does only
measurements~\cite{MF13,Mor14}, where the client uses weak coherent
pulses \cite{DKL12}, or the server uses different types of
computational resource states \cite{MDK10} have been studied. A DQC
protocol for continuous\-/variable quantum computation has been
proposed~\cite{Mor12}, as well as protocols in the circuit
\cite{GMMR13} and ancilla\-/driven \cite{SKM13} quantum computation
models. To improve the efficiency of these protocols, fault tolerant
computation has been directly embedded in
them~\cite{MF12,CMK13}. Alternatives which minimize the communication
complexity between the client and server have also been
studied~\cite{GMMR13,MPDF13}. Fisher et al.~\cite{FBSYLPJR14} have
investigated the related problem of quantum computation on encrypted
data, in which the computation is public and only the input-output are
to be kept secret.

Subsequent to this work, Morimae and Koshiba~\cite{MK13} gave a direct
composable security proof for the protocol from~\cite{Mor14}. They
obtain tighter bounds on the probability of failure than what one can
obtain using the generic reduction from local criteria proven in this
work.

The prospects of delegated quantum computation with suitable security
properties go beyond the purpose of solving computational problems for
clients.  In \cite{ABE10,AV13} verifiable quantum computation has been
linked to quantum complexity theory, and to the fundamental problem of
the feasibility of falsifying quantum
mechanics~\cite{Vaz07}. Reichardt et al.~\cite{RUV13} use an
alternative model of DQC with two non\-/communicating but entangled
servers to achieve verifiable quantum dynamics, and from this they
also prove that $\QIP=\MIP^*$. The privacy properties of secure DQC
have also been exploited in \cite{MS10}, where DQC is suggested as a
component of the verification step of unforgeable quantum coins.

It is worth mentioning that the questions of secure delegated
computation have initially been addressed in the context of classical
client\-/server scenarios. Abadi, Feigenbaum and Killian~\cite{AFK87}
considered the problem of ``computing with encrypted data'', where for
a function $f$, an instance $x$ can be efficiently encrypted into
$z=E_k(x)$ in such a way that the client can recover $f(x)$
efficiently from $k$ and $f(z)$ computed by the server. There they
showed that no \NP\=/hard function can be computed while maintaining
information\-/theoretic privacy, unless the polynomial hierarchy collapses at
the third level~\cite{AFK87}.  

A related, but distinct branch of research into the problem of
securely delegating difficult and time\-/consuming computations was
also studied in the framework of (computationally secure) public\-/key
cryptography, essentially from its very beginnings~\cite{RAD78}. Even
in this setting, this problem known as fully homomorphic encryption,
was only solved recently~\cite{Gen09}. Though the goal of the fully
homomorphic encryption program was to achieve delegated computation in
which the communication between the server and the client is
independent from the size of the desired computation. In contrast, in
all DQC proposals, the communication is essentially proportional to
the computation size; the client is however limited to operations
which are not sufficient for performing the desired computation
efficiently.\footnote{The client cannot perform the computation in
  polynomial time, assuming $\BQP \neq \BPP$.}

\subsection{Structure of this paper}
\label{sec:intro.structure}

In \secref{sec:ac} we introduce the AC framework that we use to model
security. In \secref{sec:systems} we then instantiate the abstract
systems from \secref{sec:ac} with the appropriate quantum systems and
metrics used in this work. In \secref{sec:dqc} we explain delegated
quantum computation, and model composable security for such
protocols. In \secref{sec:verif} we show that composable verifiability
(which encompasses blindness) is equivalent to the distance between
the real protocol and some ideal map that simultaneously provides both
local\-/blindness and local\-/verifiability. This map is however still
more elaborate than local criteria used in previous works. In
\secref{sec:standalone} we break this map down into individual notions
of local\-/blindness and independent local\-/verifiability, and prove
that these are sufficient to achieve security. In \secref{sec:blind}
we prove that some existing protocols are composably blind, in
particular, that of Broadbent, Kashefi and Fitzsimons~\cite{BFK09}.


\section{Abstract cryptography}
\label{sec:ac}

\subsection{Overview}
\label{sec:ac.overview}

To model security we use Maurer and Renner's~\cite{MR11} Abstract
Cryptography (AC) framework (for a more detailed introduction to
AC, we refer to \cite{PR14}). The traditional approach to defining
security can be seen as \emph{bottom\-/up}. One first defines (at a
low level) a computational model (e.g., a Turing machine or a
circuit). Based on this, the concept of an algorithm for the model and
a communication model (e.g., based on tapes) are defined. After this,
notions of complexity, efficiency, and finally the security of a
cryptosystem can be defined. The AC framework uses a \emph{top\-/down}
approach: in order to state definitions and develop a theory, one
starts from the other end, the highest possible level of abstraction
\--- the composition of abstract systems \--- and proceeds downwards,
introducing in each new lower level only the minimal necessary
specializations.

To clarify this point further, one may consider an example from
mathematics, that of group theory and the specialized problem of
matrix multiplication. In the bottom\-/up approach, one would start
explaining how matrices are multiplied, and then based on this find
properties of the matrix multiplication. In contrast to this, the AC
approach would correspond to first defining the (abstract)
multiplication group and prove theorems already on this level. The
matrix multiplication would then be introduced as a special case of
the multiplicative group, for which, naturally, all the theorems
proven on the group\-/theory level also hold.

On a high level of abstraction, a cryptographic protocol can be viewed
as (approximately) constructing some resource $\aS$ out of other
resources $\aR$. For example, a one-time pad constructs a secure
channel out of a secret key and an authentic channel; a quantum key
distribution protocol constructs an almost perfect shared secret key
out of a classical authentic channel and an insecure quantum
channel. If some protocol $\pi$ uses a resource $\aR$ to construct a
resource $\eps$\=/close to $\aS$, we write
\begin{equation} \label{eq:cc} \aR \xrightarrow{\pi,\eps} \aS.
\end{equation}
For the construction to be composable, we need the following
conditions fulfilled:
\begin{align*}
\aR \xrightarrow{\pi,\eps} \aS\ \textup{and}\ \aS
  \xrightarrow{\pi',\eps'} \aT & \implies \aR
  \xrightarrow{\pi' \circ \pi,\eps+\eps'} \aT \\
\aR \xrightarrow{\pi,\eps} \aS \ \textup{and}\  \aR'
  \xrightarrow{\pi',\eps'} \aS' & \implies \aR \| \aR'
  \xrightarrow{\pi|\pi',\eps+\eps'} \aS \| \aS'
\end{align*}
where $\aR\|\aR'$ is a parallel composition of resources, and $\pi'
\circ \pi$ and $\pi|\pi'$ are sequential and parallel composition of
protocols, respectively.

In \secref{sec:ac.security} we provide a security definition
which satisfies these conditions. Intuitively, the resource $\aR$
along with the protocol $\pi$ are part of the \emph{real} or
\emph{concrete} world, and the resource $\aS$ is some \emph{ideal
  abstraction} of the resource we want to build. \eqnref{eq:cc} is
then satisfied if an adversary could, in an ideal world where the
ideal resource is available, achieve anything that she could achieve
in the real world. This argument involves, as a thought experiment,
simulator systems which transform the ideal resource into the real
world system consisting of the real resource and the protocol.

\subsection{Resources, converters and distinguishers}
\label{sec:ac.resources}

In this section we define (on a high level of abstraction) the
elements present in \eqnref{eq:cc}, namely resources $\aR,\aS$, a
protocol $\pi$, and a pseudo\-/metric allowing us to define the
failure measure $\eps$.

Depending on what model of computing is instantiated at a lower level,
a resource can be modeled as a random system in the classical
case~\cite{Mau02,MPR07}, or, if the underlying system is quantum, as a
sequence of CPTP maps with internal memory (e.g., quantum
strategies~\cite{GW07} and combs~\cite{CDP09}).\footnote{In
  \secref{sec:systems} we define two-party protocols and quantum
  metrics on this level.} However, in order to define the security of
a protocol, it is not necessary to go down to this level of detail, a
resource can be modeled in more abstract terms.\footnote{In particular,
  on this level of abstraction it is not relevant whether the
  underlying system is classical or quantum.} A resource is an
(abstract) system with interfaces specified by a set $\cI$ (e.g., $\cI
= \{A,B,E\}$). Each interface $i \in \cI$ is accessible to a user $i$
and provides her or him with certain functionalities. Furthermore, a
dishonest user might have access to more functionalities than an
honest one, and these should be clearly marked as such (e.g., a filter
covers these functionalities for an honest player, and a dishonest
user removes the filter to access them). We call these
\emph{guaranteed} and \emph{filtered} functionalities. For example, a
key distribution resource is often modeled as a resource which either
produces a secret key or an error flag.\footnote{This is the best one
  can achieve in certain settings, e.g., quantum key distribution,
  since an adversary can cut the communication channels and prevent a
  key from being generated.} This resource has no guaranteed
functionalities at Eve's interface, but may provide her with the
filtered functionality of preventing a key being generated. Alice's
interface guarantees that she gets a secret key (or an error flag),
but it may also provide her with the filtered functionality of
choosing what key is generated.

A protocol $\pi = \{\pi_i\}_{i \in \cI}$ is a set of \emph{converters}
$\pi_i$, indexed by the set of interfaces $\cI$. A converter is an
(abstract) system with only two interfaces, an \emph{outside}
interface and an \emph{inside} interface. The outside interface is
connected to the outside world, it receives the inputs and produces
the outputs. The inside interface is connected to the resources used.

In \figref{fig:otp.real.cor} we illustrate this by connecting a
one-time pad protocol to a resource $\aR$ consisting of a secret key
and an authentic channel. Let $\pi=(\pi_A,\pi_B,\pi_E)$ be a one-time
pad protocol, and $\pi_A$ be Alice's part of the protocol: $\pi_A$ is
connected at the inner interface to a resource generating a secret key
and to an authentic channel (for this example, we assume that neither
the ideal key nor the authentic channel produce an error, they both
always generate a key and deliver the message, respectively), both of
which we combine together as the resource $\aR$. At the outer
interface it receives some message $x$, it gets a key $k$ from the key
resource, and sends $x \xor k$ down the authentic channel. Bob's part
of the protocol $\pi_B$ receives $y$ from the authentic channel and
$k$ from the key resource at its inner interface, and outputs $y \xor
k$ at the outer interface. Note that the protocol also specifies an
honest behavior for Eve, $\pi_E$, which consists in not listening to
the communication channel, i.e., it is a converter with no
functionalities at the outer interface and which blocks the leaks from
the authentic channel at the inner interface.

\begin{figure}[htb]
\begin{centering}

\begin{tikzpicture}\small

\def\t{4.913} 
\def\u{3.14} 
\def\v{.75}
\def\vv{2.5}
\def\ev{-3.4}

\node[pnode] (a1) at (-\u,\v) {};
\node[pnode] (a2) at (-\u,0) {};
\node[pnode] (a3) at (-\u,-\v) {};
\node[protocol] (a) at (-\u,0) {};
\node (atext) at (-\u,.3) {\footnotesize $y = x \xor k$};
\node[yshift=-2,above right] at (a.north west) {$\pi_A$};
\node (alice) at (-\t,0) {Alice};

\node[pnode] (b1) at (\u,\v) {};
\node[pnode] (b2) at (\u,0) {};
\node[pnode] (b3) at (\u,-\v) {};
\node[protocol] (b) at (\u,0) {};
\node (btext) at (\u,.3) {\footnotesize $x = y \xor k$};
\node[yshift=-2,above right] at (b.north west) {$\pi_B$};
\node (bob) at (\t,0) {Bob};

\node[largeResource] (R) at (0,.125) {};
\node[yshift=-2,above right] at (R.north west) {$\aR$};
\node[thinResource] (keyBox) at (0,\v) {};
\node[draw] (key) at (0,\v) {key};
\node[yshift=-2,above right] at (keyBox.north west) {\footnotesize Secret key};
\node[thinResource] (channel) at (0,-\v) {};
\node[yshift=-1.5,above right] at (channel.north west) {\footnotesize
  Authentic channel};
\node[protocolLong] (e) at (0,-\vv) {};
\node[xshift=2,below left] at (e.north west) {$\pi_E$};
\node (junc) at (e |- a3) {};
\node (eleft) at (-.5,-\vv) {};
\node (eright) at (.5,-\vv) {};
\draw[ultra thick] (eleft.center) to (eright.center);
\node (eve) at (0,\ev) {Eve};

\draw[sArrow] (key) to node[auto,swap,pos=.3] {$k$} (a1);
\draw[sArrow] (key) to node[auto,pos=.3] {$k$} (b1);

\draw[sArrow] (alice) to node[auto,pos=.35] {$x$} (a2);
\draw[sArrow] (b2) to node[auto,pos=.6] {$x$} (bob);

\draw[sArrow] (a3) to node[pos=.1,auto] {$y$} node[pos=.9,auto] {$y$} (b3);
\draw[sArrow] (junc.center) to node[pos=.59,auto] {$y$} (e.center);

\end{tikzpicture}

\end{centering}
\caption[One-time pad]{\label{fig:otp.real.cor}The concrete setting of the one-time
  pad with Eve's honest protocol $\pi_E$. Alice has access to the left
  interface, Bob to the right interface and Eve to the lower
  interface. The converters $(\pi_A,\pi_B,\pi_E)$ of the one-time pad
  protocol are connected to the resource $\aR$ consisting of a secret
  key and an authentic channel.}
\end{figure}
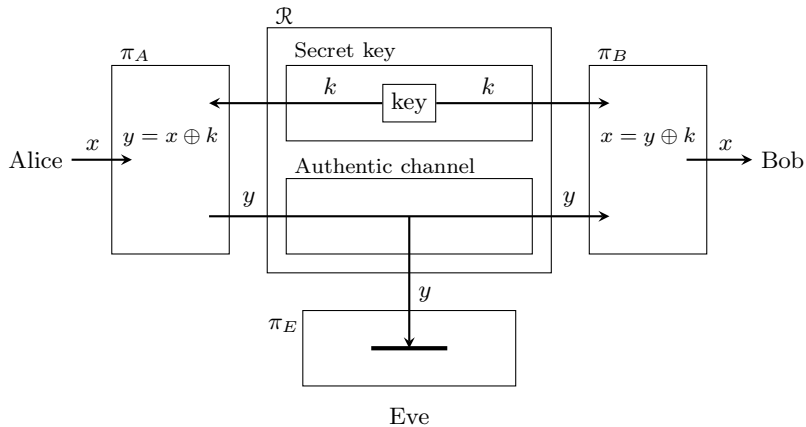

Converters connected to resources build new resources with the same
interface set, and we write either $\pi_i \aR$ or $\aR\pi_i$ to denote
the new resource with the converter $\pi_i$ connected at the interface
$i$.\footnote{There is no mathematical difference between $\pi_i \aR$
  and $\aR\pi_i$. It sometimes simplifies the notation to have the
  converters for some players written on the right of the resource and
  the ones for other players on the left, instead of all on the same
  side, hence the two notations.}

Filters, which cover the cheating interface when a player is honest,
can also be modeled as converters.

To measure how close two resources are, we define a pseudo\=/metric on
the space of resources. We do this with the help of a
\emph{distinguisher}. For $n$\=/interface resources a distinguisher
$\aD$ is a system with $n+1$ interfaces, where $n$ interfaces connect
to the interfaces of a resource $\aR$ and the other (outside)
interface outputs a bit. For a class of distinguishers $\mathfrak{D}$,
the induced pseudo\-/metric, the distinguishing advantage,
is \[d(\aR,\aS) \coloneqq \max_{\aD \in \mathfrak{D}} \Pr\left[\aD\aR =
  1\right] - \Pr\left[\aD\aS = 1\right],\] where $\aD\aR$ is the
binary random variable corresponding to $\aD$ connected to
$\aR$.\footnote{In this work we study information\-/theoretic security,
  and therefore the only class of distinguishers that we consider is
  the set of all distinguishers.} If $d(\aR,\aS) \leq \eps$, we say
that the two resources are $\eps$-close and sometimes write $\aR
\close{\eps} \aS$; or $\aR = \aS$ if $\eps = 0$.

\subsection{Security}
\label{sec:ac.security}

We now have introduced all the notions used in the generic security
definition:\footnote{In \cite{MR11} this definition is given on a
  higher level of abstraction. However for the particular case of
  filtered resources, \defref{def:cc} is equivalent.}
\begin{deff}[See \cite{MR11}]
\label{def:cc}
Let $\aR_\phi = (\aR,\phi)$ and $\aS_\psi=(\aS,\psi)$ be pairs of a
resource ($\aR$ and $\aS$) with interfaces $\cI$ and a filter ($\phi$
and $\psi$). We say that a protocol $\pi$ (securely) constructs
$\aS_\psi$ out of $\aR_\phi$ within $\eps$, and write $\aR_\phi
\xrightarrow{\pi,\eps} \aS_\psi$, if there exist converters $\sigma =
\{\sigma_i\}_{i \in \cI}$ \--- which we call simulators \--- such
that,
\begin{equation}
  \label{eq:cc.def} \forall \cP \subseteq \cI, \quad d(\pi_\cP \phi_\cP\aR,\sigma_{\cI
    \setminus \cP}\psi_\cP\cS) \leq \eps,\end{equation}
where for $x = \{x_i\}_{i \in \cI}$, $x_{\cP} \coloneqq \{x_i\}_{i \in \cP}$.\end{deff}

We illustrate this definition in the case of the one-time pad. In this
example, we wish to construct a secure channel $\aS$, which is
depicted in \figref{fig:otp.ideal.cor} and defined as follows (for
simplicity, we assume that Alice and Bob are always honest, and ignore
their filtered functionalities): $\aS$ takes a message $x$ at the
$A$\=/interface, leaks the message length $|x|$ at the
$E$\=/interface, and outputs $x$ at the $B$\=/interface. This resource
captures the desired notion of a secure channel, because it only leaks
the message size, and does not provide the adversary with any
functionality to falsify the message. We model explicitly that the
message size leak at the $E$\=/interface is not a guaranteed
functionality by depicting it in gray in
\figref{fig:otp.ideal.cor}. We additionally draw the filter converter
$\psi_E$, which covers the cheating interface and can be removed by a
dishonest player. $\psi_E$ has no functionalities at the outer interface,
and blocks this message size leak at the inner interface. In the
general case, these filters can be defined for all
interfaces.\footnote{We only denote Eve's filter explicitly in the
  following, since Alice and Bob's filters are trivial (the
  identity).}

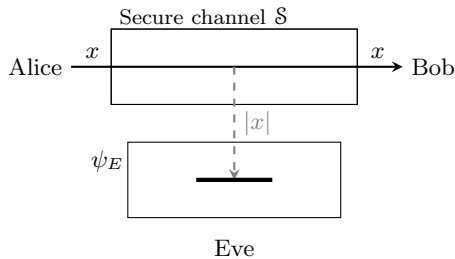
\begin{figure}[htb]
\begin{centering}

\begin{tikzpicture}\small

\def\t{2.618} 
\def\u{-.75}
\def\v{.75}
\def\w{-2.45} 
\def\ev{-1.65} 

\node[thinResource] (channel) at (0,\v) {};
\node[yshift=-1.5,above right] at (channel.north west) {{\footnotesize Secure channel} $\aS$};
\node[thinResource] (channel) at (0,\v) {};
\node (alice) at (-\t,\v) {Alice};
\node (bob) at (\t,\v) {Bob};

\node[protocolLong] (filter) at (0,\u) {};
\node[xshift=2,below left] at (filter.north west) {$\psi_E$};

\draw[sArrow] (alice) to node[pos=.065,auto] {$x$}
node[pos=.925,auto] {$x$} (bob);
\draw[gArrow,dashed] (0,\v) to node[pos=.5,auto] {$|x|$} (filter.center);
\node (eleft) at (-.5,\u) {};
\node (eright) at (.5,\u) {};
\draw[ultra thick] (eleft.center) to (eright.center);
\node (eve) at (0,\ev) {Eve};

\end{tikzpicture}

\end{centering}
\caption[Secure channel]{\label{fig:otp.ideal.cor}A secure channel from Alice to
  Bob. Alice has access to the left interface, Bob to the right
  interface and Eve to the lower interface. A filter $\psi_E$ covers
  Eve's cheating functionality.}
\end{figure}

The \emph{correctness} of the protocol $\pi$ is captured by measuring
the distance between $\pi\aR$, the combination of the entire honest
protocol with the resources (\figref{fig:otp.real.cor}), and
$\psi\aS$, the ideal resource with all filtered functionalities
obstructed (\figref{fig:otp.ideal.cor}). In the case of the one-time
pad, we have $d(\pi_A\pi_B\pi_E\aR,\psi_E\aS) = 0$: since the
resources $\pi_A\pi_B\pi_E\aR$ and $\psi_E\aS$ both simply take a
message $x$ as input at the $A$\=/interface and output the same
message at the $B$\=/interface, no distinguisher can notice a
difference.

If a player $i$ cheats, she does not (necessarily) follow her protocol
$\pi_i$, but can interact arbitrarily with her interface. We thus
remove the corresponding protocol converters from the real setting to
model the resulting resource, which we depict for the one-time pad in
\figref{fig:otp.real.sec}. \emph{Security} of the protocol in the
presence of a cheating party $i$ is achieved if this player is not
able to accomplish more than what is allowed by her interface of the
ideal resource with the filter removed. This is the case if there
exists a simulator converter $\sigma_i$, independent from the cheating
strategy, that, when plugged into the $i$\=/interface of the ideal
resource $\aS$, can convert between the interaction with the corrupt
player (or distinguisher) and the filtered functionalities of the
resource, such that the real and ideal worlds are
indistinguishable. For example, in the case of the one-time pad and a
dishonest Eve, a cipher $y$ is leaked at the $E$\=/interface, whereas
in the ideal setting, only the message length is leaked. The simulator
$\sigma_E$ therefore must recreate a cipher given the message
length. It does this by simply generating a random string $y$ of the
corresponding length and outputting it at its outer interface. This is
illustrated in \figref{fig:otp.ideal.sec}. It is not hard to verify
that with this simulator, $d(\pi_A\pi_B\aR,\sigma_E\aS) = 0$, since
the resources $\pi_A\pi_B\aR$ and $\sigma_E\aS$ both take a message
$x$ at their $A$\=/interface, which they output at their
$B$\=/interface, and output a completely random string of the same
length at their $E$\=/interface.

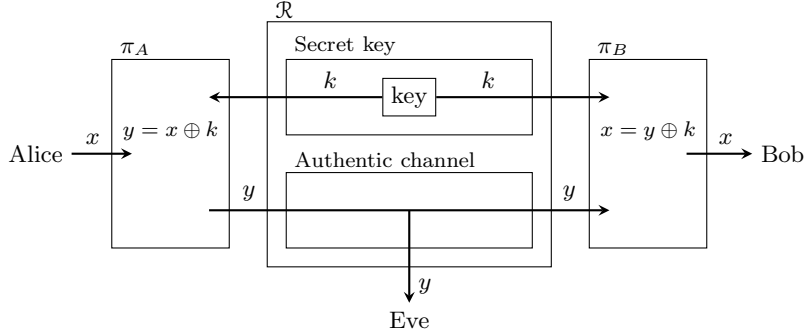
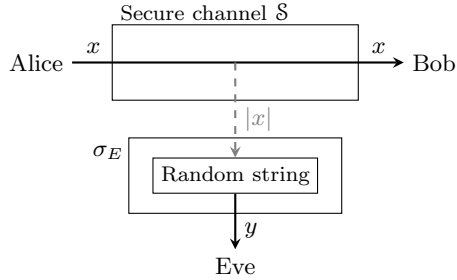
\begin{figure}[htb]
\begin{subfigure}[b]{\textwidth}
\begin{centering}

\begin{tikzpicture}\small

\def\t{4.913} 
\def\u{3.14} 
\def\v{.75}
\def\vv{2.2}

\node[pnode] (a1) at (-\u,\v) {};
\node[pnode] (a2) at (-\u,0) {};
\node[pnode] (a3) at (-\u,-\v) {};
\node[protocol] (a) at (-\u,0) {};
\node (atext) at (-\u,.3) {\footnotesize $y = x \xor k$};
\node[yshift=-2,above right] at (a.north west) {$\pi_A$};
\node (alice) at (-\t,0) {Alice};

\node[pnode] (b1) at (\u,\v) {};
\node[pnode] (b2) at (\u,0) {};
\node[pnode] (b3) at (\u,-\v) {};
\node[protocol] (b) at (\u,0) {};
\node (btext) at (\u,.3) {\footnotesize $x = y \xor k$};
\node[yshift=-2,above right] at (b.north west) {$\pi_B$};
\node (bob) at (\t,0) {Bob};

\node[largeResource] (R) at (0,.125) {};
\node[yshift=-2,above right] at (R.north west) {$\aR$};
\node[thinResource] (keyBox) at (0,\v) {};
\node[draw] (key) at (0,\v) {key};
\node[yshift=-2,above right] at (keyBox.north west) {\footnotesize Secret key};
\node[thinResource] (channel) at (0,-\v) {};
\node[yshift=-1.5,above right] at (channel.north west) {\footnotesize
  Authentic channel};
\node (eve) at (0,-\vv) {Eve};
\node (junc) at (eve |- a3) {};

\draw[sArrow] (key) to node[auto,swap,pos=.3] {$k$} (a1);
\draw[sArrow] (key) to node[auto,pos=.3] {$k$} (b1);

\draw[sArrow] (alice) to node[auto,pos=.35] {$x$} (a2);
\draw[sArrow] (b2) to node[auto,pos=.6] {$x$} (bob);

\draw[sArrow] (a3) to node[pos=.1,auto] {$y$} node[pos=.9,auto] {$y$} (b3);
\draw[sArrow] (junc.center) to node[pos=.82,auto] {$y$} (eve);

\end{tikzpicture}

\end{centering}
\caption{\label{fig:otp.real.sec}The concrete resource resulting from honest Alice and Bob
  running their one-time pad protocols $(\pi_A,\pi_B)$ with a secret
  key and authentic channel.}
\end{subfigure}

\vspace{12pt}

\begin{subfigure}[b]{\textwidth}
\begin{centering}

\begin{tikzpicture}\small

\def\t{2.618} 
\def\u{-.75}
\def\v{.75}
\def\w{-1.95} 

\node[thinResource] (channel) at (0,\v) {};
\node[yshift=-1.5,above right] at (channel.north west) {{\footnotesize
  Secure channel} $\aS$};
\node (alice) at (-\t,\v) {Alice};
\node (bob) at (\t,\v) {Bob};

\node[protocolLong] (sim) at (0,\u) {};
\node[xshift=1.5,below left] at (sim.north west) {$\sigma_E$};
\node[draw] (rand) at (0,\u) {\footnotesize Random string};

\draw[sArrow] (alice) to node[pos=.065,auto] {$x$}
node[pos=.925,auto] {$x$} (bob);
\draw[gArrow,dashed] (0,\v) to node[pos=.6,auto] {$|x|$} (rand);

\node (eve) at (0,\w) {Eve};
\draw[sArrow] (rand) to node[pos=.65,auto] {$y$} (eve);

\end{tikzpicture}

\end{centering}
\caption{\label{fig:otp.ideal.sec}The ideal resource $\aS$ constructed
  by the one-time pad for an honest Alice and Bob, and a simulator
  $\sigma_E$ plugged into Eve's interface.}
\end{subfigure}
\caption[Real and ideal settings for the one-time
pad]{\label{fig:otp.sec}The real and ideal settings for the one-time
  pad with a cheating Eve. Alice has access to the left interface, Bob
  to the right interface and Eve to the lower interface. Since these
  resources are indistinguishable, the one-time pad provides perfect
  security.}
\end{figure}

\defref{def:cc} requires $2^n$ inequalities to be satisfied in a model
with $n$ players, i.e., one for every possible subset of dishonest
players. In practice however, if we are only interested in modeling
security when a given set of players is known to always be honest \---
e.g., Alice and Bob are honest in the one-time pad example \--- then
it is sufficient to consider only the corresponding inequalities from
\eqnref{eq:cc.def}. This is equivalent to giving those players
arbitrary filtered functionalities, and reflects the fact that we do
not place any restrictions on what these players might achieve, were
they to be dishonest.

\begin{rem}
  \label{rem:cc}
  Abstract cryptography (AC) differs from universal composability
  (UC)~\cite{Can13,Unr10} in many conceptual and mathematical ways. In
  particular, the AC requirement that there exist distinct simulators
  at each interface instead of merging all dishonest players into one
  entity make it strictly more powerful than UC: this allows dishonest
  players with restricted cooperation to be modeled as a feature of
  the ideal resource, and thus directly capture notions such as
  coercibility~\cite{MR11}.

  However, in the special case of one dishonest player,
  \eqnref{eq:cc.def} is equivalent to what one obtains by modeling the
  same problem with UC. Since the rest of this work deals with
  delegated quantum computation, a two-party protocol with one
  dishonest player, the same results could have been obtained using
  the UC framework.
\end{rem}

\section{Quantum systems}
\label{sec:systems}

In \secref{sec:ac} resources and converters were introduced as
abstract systems. Here we model them explicitly for the special case
of two-party protocols considered in the rest of the work. In
\secref{sec:systems.notation} we first briefly define the notation and
some basic concepts that we use.\footnote{For a more detailed
  introduction to quantum information theory we refer to
  \cite{NC00,Wat11}.} In \secref{sec:systems.two} we then model
two-party protocols. And finally in \secref{sec:systems.dist} we
define several metrics which correspond to the distinguishing
advantage for specific resources.

\subsection{Notation and basic concepts}
\label{sec:systems.notation}

\hilbert\ always denotes a finite\-/dimensional Hilbert space. We denote
by $\lo[B]{A}$ the set of linear operators from $\hilbert_A$ to
$\hilbert_B$, by $\lo{}$ the set of linear operators from
$\hilbert$ to itself, and by $\po{}$ the
subset of positive semi-definite operators. We define the set of
normalized quantum states $\no{} \coloneqq \{ \rho \in \po{} : \tr
\rho = 1\}$ and the set of subnormalized quantum states $\sno{}
\coloneqq \{ \rho \in \po{}: \tr \rho \leq 1\}$. We write
$\hilbert_{AB} = \hilbert_A \tensor \hilbert_B$ for a bipartite
quantum system and $\rho_{AB} \in \sno{AB}$ for a bipartite quantum
state. $\rho_A = \trace[B]{\rho_{AB}}$ and $\rho_B =
\trace[A]{\rho_{AB}}$ denote the corresponding reduced density
operators.

The set of feasible maps between two systems $A$ and $B$ is the set of
all completely positive, trace\-/preserving (CPTP) maps $\cE : \lo{A}
\to \lo{B}$. By the Kraus representation, such a map can always be
given by a set of linear operators $\{E_k \in \lo[B]{A}\}_k$ with
$\sum_k \hconj{E}_kE_k = \I_A$. We then have $\cE(\rho) = \sum_k E_k
\rho \hconj{E}_k$. We also consider trace non\-/increasing maps \---
in particular, to describe the evolution of a system conditioned on a
specific measurement outcome \--- i.e., maps with operators $E_k$ such
that $\sum_k \hconj{E}_kE_k \leq \I_A$. Though when unspecified, we
always mean trace\-/preserving maps. For a quantum state $\rho \in
\sno{AC}$ and a map $\cE : \lo{A} \to \lo{B}$, $\cE(\rho)$ is
shorthand for $\left(\cE \tensor \id_C\right) (\rho)$, where $\id_C$
is the identity on system $C$.

Throughout this paper we mostly use the standard notation for common
quantum gates, for instance $X$ and $Z$ denote the Pauli-$X$ and
Pauli-$Z$ operators. We will additionally often refer to the the
parametrized phase gate $Z_{\theta} = \proj{0}+ e^{i \theta} \proj{1}$,
and the two-qubit controlled-$Z$ gate $\ctrlZ =
\proj{00}+\proj{01}+\proj{10} - \proj{11}$. 

\subsection{Two-party protocols}
\label{sec:systems.two}

A two-party protocol can in general be modeled by a sequence of CPTP
maps $\{\cE_i : \lo{AC} \to \lo{AC} \}_i$ and $\{\cF_i : \lo{CB} \to
\lo{CB} \}_i$, where $A$ and $B$ are Alice and Bob's registers, and
$C$ represents a communication channel.\footnote{One could consider a
  more general two-party setting, where the players have access to
  other resources than a channel, e.g., public randomness. But since
  in the rest of this work we are interested only in protocols where
  the players have no other resource than a channel, we also consider
  only this case here.} Initially Alice and Bob place their inputs in
their registers, and the channel $C$ is in some fixed state
$\sket{0}$. The players then apply successively their maps to their
respective registers and the channel. For example, in the first round
Alice applies $\cE_1$ to the joint system $AC$, and sends $C$ to Bob,
who applies $\cF_1$ to $CB$, and returns $C$ to Alice. Then she
applies $\cE_2$, etc.

In the AC terminology introduced in \secref{sec:ac}, the messages sent
on the channel $C$ correspond to messages leaving a converter at the
inner interface and being sent through a channel resource $\aR$ to the
other player. The inputs are initially received by the converters at
the outer interfaces, and the final contents of the $A$ and $B$
registers is output at the outer interface once the last map of the
protocol has been applied. This is illustrated in
\figref{fig:prot.real}.

\begin{figure}[htb]
\begin{centering}

\begin{tikzpicture}[opnode/.style={minimum width=1.05cm,minimum height=.6cm}]
\small

\def\t{4.1}
\def\u{2.6}
\def\v{.9}

\node[draw,opnode] (a1) at (-\u,0) {$\cE_1$};
\node[draw,opnode] (a2) at (-\u,-\v) {$\cE_2$};
\node[opnode] (a3) at (-\u,-2*\v) {};
\node[opnode] (a35) at (-\u,-2.12*\v) {$\vdots$};
\node[opnode] (a4) at (-\u,-2.5*\v) {};
\node[draw,opnode] (a5) at (-\u,-3.5*\v) {$\cE_N$};
\node[draw,inner sep=.25cm,fit=(a1)(a5)] (a) {};
\node[yshift=-2,above right] at (a.north west) {$\pi_A$};
\node (alice1) at (-\t,0) {};
\node (alice2) at (-\t,-3.5*\v) {};
\node (alice) at (-\t-.4,-1.75*\v) {Alice};

\node[draw,opnode] (b1) at (\u,0) {$\cF_1$};
\node[draw,opnode] (b2) at (\u,-\v) {$\cF_2$};
\node[opnode] (b3) at (\u,-2*\v) {};
\node[opnode] (b35) at (\u,-2.12*\v) {$\vdots$};
\node[opnode] (b4) at (\u,-2.5*\v) {};
\node[draw,opnode] (b5) at (\u,-3.5*\v) {$\cF_{N}$};
\node[draw,inner sep=.25cm,fit=(b1)(b5)] (b) {};
\node[yshift=-2,above right] at (b.north west) {$\pi_B$};
\node (bob1) at (\t,0) {};
\node (bob2) at (\t,-3.5*\v) {};
\node (bob) at (\t+.4,-1.75*\v) {Bob};

\node[opnode] (c1) at (0,0) {};
\node[opnode] (c2) at (0,-1.5*\v) {};
\node[opnode] (c35) at (0,-2.12*\v) {$\vdots$};
\node[opnode] (c4) at (0,-3*\v) {};
\node[opnode] (c5) at (0,-3.5*\v) {};
\node[draw,minimum width=2.2cm,inner sep=.25cm,fit=(c1)(c5)] (c) {};
\node[yshift=-2,above right] at (c.north west) {$\aR$};

\draw[sArrow] (alice1.center) to node[auto,pos=.4] {$\psi_A$} (a1);
\draw[sArrow] (bob1.center) to node[auto,pos=.4,swap] {$\psi_B$} (b1);
\draw[sArrow] (a1) to (b1);
\draw[sArrow] (a1) to (a2);
\draw[sArrow] (b1) to (a2);
\draw[sArrow] (b1) to (b2);
\draw[sArrow] (a2) to (b2);
\draw[sArrow] (a2) to (a3);
\draw[sArrow] (b2) to (c2);
\draw[sArrow] (b2) to (b3);
\draw[sArrow] (c4) to (a5);
\draw[sArrow] (a4) to (a5);
\draw[sArrow] (b4) to (b5);
\draw[sArrow] (a5) to (b5);
\draw[sArrow] (a5) to node[auto,swap,pos=.6] {$\rho_A$} (alice2.center);
\draw[sArrow] (b5) to node[auto,pos=.6] {$\rho_B$} (bob2.center);

\end{tikzpicture}

\end{centering}
\caption[Generic two-party protocol]{\label{fig:prot.real}A generic
  two-party protocol. Alice has access to the left interface and Bob
  to the right interface. The protocol $(\pi_A,\pi_B)$ consists in
  sequences of maps. The channel resource $\aR$ simply transmits the
  messages between the players.}
\end{figure}
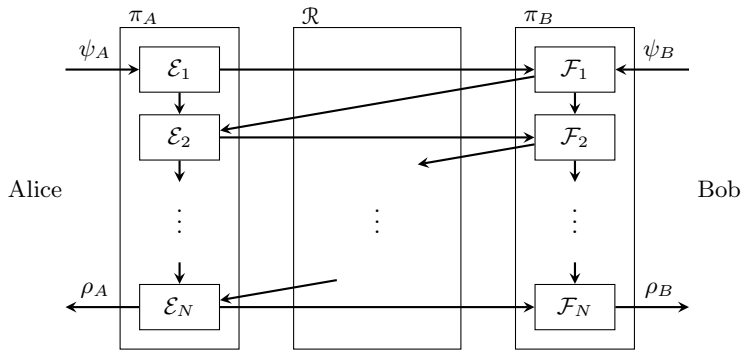

For a protocol with $N$ rounds, the resource $\pi_i\aR$, corresponding
to one of the players' protocol plugged into the channel resource
$\aR$, has been called a \emph{quantum strategy} by Gutoski and
Watrous~\cite{GW07,Gut12} and a \emph{quantum $N$\=/comb} by
Chiribella, D'Ariano and Perinotti~\cite{CDP09}. In particular, these
authors derived independently a concise representation of
combs/strategies in terms of the Choi\-/Jamio\l{}kowski
isomorphism. They also define the appropriate distance measure between
two combs/strategies, corresponding to the optimal distinguishing
advantage, which we sketch in the next section.

\subsection{Distance measures}
\label{sec:systems.dist}

The trace distance between two states $\rho$ and $\sigma$ is given by
$D(\rho,\sigma) = \frac{1}{2} \trnorm{\rho-\sigma}$, where
$\trnorm{\cdot}$ denotes the trace norm and is defined as $\trnorm{A}
\coloneqq \tr \sqrt{\hconj{A}A}$.  If $D(\rho,\sigma) \leq \eps$, we
say that the two states are $\eps$-close and often write $\rho
\close{\eps} \sigma$. This corresponds to the distinguishing advantage
between two resources $\aR$ and $\aS$, which take no input and produce
$\rho$ and $\sigma$, respectively, as output: the probability of a
distinguisher guessing correctly whether he holds $\aR$ or $\aS$ is
exactly $\frac{1}{2}+\frac{1}{2}D(\rho,\sigma)$. In \appendixref{app:dist}
we define the generalized trace distance and the purified distance, which
are more appropriate for characterizing the distance between
subnormalized states.

Another common metric which corresponds to the distinguishing
advantage between resources of a certain type is the diamond norm. If
the resources $\aR$ and $\aS$ take an input $\rho \in \no{A}$ and
produce an output $\sigma \in \no{B}$, the distinguishing advantage
between these resources is the diamond distance between the correspond
maps $\cE,\cF : \lo{A} \to \lo{B}$. A distinguisher can generate a
state $\rho_{AR}$, input the $A$ part to the resource, and try to
distinguish between the resulting states $\cE(\rho_{AR})$ and
$\cF(\rho_{AR})$. We have $d(\aR,\aS) = \diamond(\cE,\cF) =
\frac{1}{2} \dianorm{\cE - \cF}$, where \[\dianorm{\Phi} \coloneqq
\max \{\trnorm{\left(\Phi \tensor \id_R\right) (\rho)} : \rho \in
\no{AR}\}\] is the diamond norm. Note that the maximum of the diamond
norm can always be achieved for a system $R$ with $\dim\hilbert_R =
\dim\hilbert_A$. Here too, we sometimes write $\cE \close{\eps} \cF$
if two maps are $\eps$-close.

If the resources considered are halves of two player protocols, say
$\pi_i\aR$ or $\pi_j\aR$, the above reasoning can be generalized for
obtaining the distinguishing advantage. The distinguisher can first
generate an initial state $\rho \in \no{AR}$ \--- which for
convenience we define as a map on no input $\rho \coloneqq \cD_0()$
\--- and input the $A$ part of the state into the resource. It
receives some output $\rho_{CR}$ from the resource, can apply some
arbitrary map $\cD_1 : \lo{CR} \to \lo{CR}$ to the state, and input
the $C$ part of the new state in the resource. Let it repeat this
procedure with different maps $\cD_i$ until the end of the protocol,
after which it holds one of two states: $\varphi_{AR}$ if it had access
to $\pi_i\aR$ and $\psi_{AR}$ if it had access to $\pi_j\aR$. The
trace distance $D(\varphi_{AR},\psi_{AR})$ defines the advantage the
distinguisher has of correctly guessing whether it was interacting
with $\pi_i\aR$ or $\pi_j\aR$, and by maximizing this over all
possible initial inputs $\rho_{AR} = \cD_0()$, and all subsequent maps
$\{\cD_{i} : \lo{CR} \to \lo{CR}\}_i$, the distinguishing advantage
between these resources becomes
\begin{equation}
\label{eq:dqc.dist}
d(\pi_i\aR,\pi_j\aR) = \max_{\{\cD_i\}_i}
D(\varphi_{AR},\psi_{AR}). \end{equation}
This has been studied by both Gutoski~\cite{Gut12} and Chiribella et
al.~\cite{CDP09}, and we refer to their work for more details.

\section{Delegated quantum computation}
\label{sec:dqc}

In the (two-party) delegated quantum computation (DQC) model, Alice
asks a server, Bob, to execute some quantum computation for her.
Intuitively, Alice plays the role of a client, and Bob the part of a
computationally more powerful server. Alice has several security
concerns. She wants the protocol to be blind, that is, she wants the
server to execute the quantum computation without learning anything
about the input other than what is unavoidable, e.g.,
an upper bound on its size, and possibly whether the output is
classical or quantum. She may also want to know if the result sent to
her by Bob is correct, which we refer to as verifiability.

In \secref{sec:dqc.model} we model the ideal resource that a DQC
protocol constructs and the structure of a generic DQC protocol. And
in \secref{sec:dqc.security} we apply the generic AC security
definition (\defref{def:cc}) to DQC.

\subsection{DQC model}
\label{sec:dqc.model}

\subsubsection{Ideal resource}
\label{sec:dqc.model.ideal}

To model the security (and correctness) of a delegated quantum computation
protocol, we need to model the ideal delegated computation resource $\aS$
that we wish to build. We start with an ideal resource that provides
blindness, and denote it $\aS^{\blind}$.

The task Alice wants to be executed is provided as an input to the
resource $\aS^{\blind}$ at the $A$\=/interface. It could be modeled
as having two parts, some quantum state $\psi_{A_1}$ and a classical
description $\Phi_{A_2}$ of some quantum operation that she wants to
apply to $\psi$, i.e., she wishes to compute $\Phi(\psi)$. This can
alternatively be seen as applying a universal computation $\cU$ to the
input $\psi_{A_1} \tensor \proj{\Phi}_{A_2}$. We adopt this view in
the remainder of this paper, and model the resource as performing some
fixed computation $\cU$ on an input $\psi_A$ that may be part quantum
and part classical.\footnoteremember{fn:classical.input}{Alternatively,
  the input can be modeled as entirely quantum, and both Alice and the
  ideal resource first measure the part of the input that should be
  classical, before executing $\pi_A$ and the universal computation
  $\cU$, respectively. This corresponds to plugging an extra
  measurement converter into the $A$\=/interfaces of both the real and
  ideal systems (that converts the quantum input into a
  classical\-/quantum input), which can only decrease the distance
  between the real and ideal systems, i.e., increase the
  security.}

  Any DQC protocol must reveal to the server an upper bound on the
  size of the computation it is required to execute. Other information
  might also be made intentionally available, such as whether the
  output of the computation is classical or quantum. Although one
  could imagine a generic DQC model in which these ``permitted leaks''
  are entangled with the rest of the input, we restrict our
  considerations to classical information, i.e., a subsystem of the
  input $\psi_A$ is classical\footnoterecall{fn:classical.input} and
  contains a string $\ell^{\psi_A} \in \{0,1\}^*$ that is copied and
  provided to the server Bob at the start of the protocol, so that he
  may set up the required resources and programs for the
  computation. Alternatively, this string can be taken to be some
  fixed publicly available information, not modeled explicitly. We do
  so in the following sections to simplify the notation, but prefer
  make it explicit in this section so as not to hide the fact that
  some information about the input is always given to the server.

The ideal resource $\aS^{\blind}$ thus takes this input $\psi_A$ at its
$A$\=/interface, and, if Bob does not activate his filtered
functionalities \--- which can be modeled by a bit $b$, set to $0$ by
default, and which a simulator $\sigma_B$ can flip to $1$ to signify
that it is activating the cheating interface \--- $\aS^{\blind}$ outputs
$\cU(\psi_A)$. This ensures both correctness and universality (in the
case where $\cU$ is a universal computation). Alternatively, $\aS^{\blind}$
can be restricted to work for inputs corresponding to a certain class
of computational problems, if we desire a construction only designed
for such a class.

If the cheating $B$\=/interface is activated, the ideal resource
outputs a copy of the string $\ell^{\psi_A}$ at this interface. Bob
also has another filtered functionality, one which allows him to
tamper with the final output. The most general operation he could
perform is to give $\aS^{\blind}$ a quantum state $\psi_B$ \--- which
could be entangled with Alice's input $\psi_A$ \--- along with the
description of some map $\cE : \lo{AB} \to \lo{A}$, and ask it to
output $\cE(\psi_{AB})$ at Alice's interface. Since $\aS^{\blind}$
only captures blindness, but says nothing about Bob's ability to
manipulate the final output, we define it to perform this operation
and output any $\cE(\psi_{AB})$ at Bob's request. This is depicted in
\figref{fig:dqc.ideal.b} with the filtered functionalities in gray.

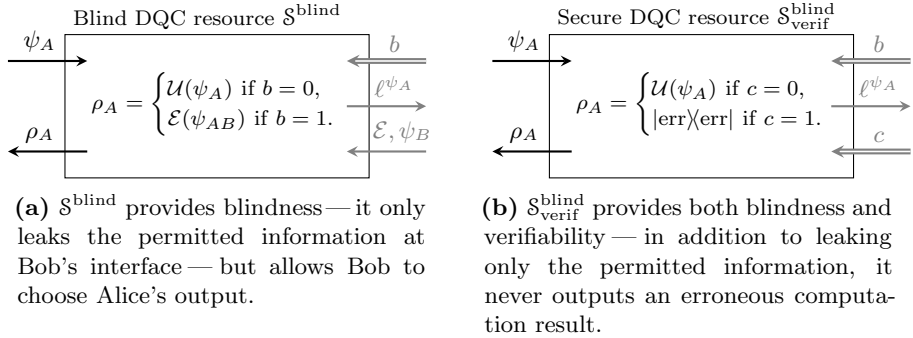
\begin{figure}[htb]
  \subcaptionbox{\label{fig:dqc.ideal.b}$\aS^{\blind}$ provides blindness
    \--- it only leaks the permitted information at Bob's interface
    \--- but allows Bob to choose Alice's output.}[.5\textwidth][l]{
\begin{tikzpicture}\small

\def\u{2.75}
\def\v{.6}

\node[brnode] (n1) at (0,\v) {};
\node[brnode] (n2) at (0,0) {};
\node[brnode] (n3) at (0,-\v) {};
\node[blocResource] (S) at (0,0) {};
\node[yshift=-1.5,above right] at (S.north west) {{\footnotesize
  Blind DQC resource} $\aS^{\blind}$};
\node[text width=3.3cm] at (0,.15) {\footnotesize \[ \rho_A = \begin{cases}
    \cU(\psi_A) \text{ if $b=0$,} \\ \cE(\psi_{AB})
    \text{ if $b=1$.} \end{cases} \]};

\node (a1) at (-\u,\v) {};
\node (a3) at (-\u,-\v) {};
\node (b1) at (\u,\v) {};
\node (b2) at (\u,0) {};
\node (b3) at (\u,-\v) {};

\draw[sArrow] (a1.center) to node[auto,pos=.4] {$\psi_A$} (n1);
\draw[sArrow] (n3) to node[auto,swap,pos=.6] {$\rho_A$} (a3.center);

\draw[gArrow,double] (b1.center) to node[auto,swap,pos=.4] {$b$} (n1);
\draw[gArrow] (n2) to node[auto,pos=.6] {$\ell^{\psi_A}$} (b2.center);
\draw[gArrow] (b3.center) to node[auto,swap,pos=.3] {$\cE,\psi_B$} (n3);

\end{tikzpicture}
} \subcaptionbox{\label{fig:dqc.ideal.v}$\aS^{\blind}_{\verif}$ provides both
  blindness and verifiability \--- in addition to leaking only the
  permitted information, it never outputs an erroneous computation
  result.}[.5\textwidth][r]{
\begin{tikzpicture}\small

\def\u{2.75}
\def\v{.6}

\node[brnode] (n1) at (0,\v) {};
\node[brnode] (n2) at (0,0) {};
\node[brnode] (n3) at (0,-\v) {};
\node[blocResource] (S) at (0,0) {};
\node[yshift=-1.5,above right] at (S.north west) {{\footnotesize
  Secure DQC resource} $\aS^{\blind}_{\verif}$};
\node[text width=3.3cm] at (0,.15) {\footnotesize \[ \rho_A = \begin{cases}
    \cU(\psi_A) \text{ if $c=0$,} \\ \proj{\err}
    \text{ if $c=1$.} \end{cases} \]};

\node (a1) at (-\u,\v) {};
\node (a3) at (-\u,-\v) {};
\node (b1) at (\u,\v) {};
\node (b2) at (\u,0) {};
\node (b3) at (\u,-\v) {};

\draw[sArrow] (a1.center) to node[auto,pos=.4] {$\psi_A$} (n1);
\draw[sArrow] (n3) to node[auto,swap,pos=.6] {$\rho_A$} (a3.center);

\draw[gArrow,double] (b1.center) to node[auto,swap,pos=.4] {$b$} (n1);
\draw[gArrow] (n2) to node[auto,pos=.6] {$\ell^{\psi_A}$} (b2.center);
\draw[gArrow,double] (b3.center) to node[auto,swap,pos=.4] {$c$} (n3);

\end{tikzpicture}
}
\caption[Ideal DQC resources]{\label{fig:dqc.ideal}Ideal DQC
  resources. The client Alice has access to the left interface, and
  the server Bob to the right interface. The double\-/lined input
  flips a bit set by default to $0$. The functionalities provided at
  Bob's interface are grayed to signify that they are accessible only
  to a cheating server. If Bob is honest, this interface is obstructed
  by a filter, which we denote by $\bot_B$ in the following.}
\end{figure}

\begin{deff}
\label{def:dqc.b}
The ideal DQC resource $\aS^{\blind}$ which provides both correctness
and blindness takes an input $\psi_A$ at Alice's interface, but no
honest input at Bob's interface. Bob's filtered interface has a
control bit $b$, set by default to $0$, which he can flip to activate
the other filtered functionalities. The resource $\aS^{\blind}$ then
outputs the permitted leak $\ell^{\psi_A}$ at Bob's interface, and accepts two
further inputs, a state $\psi_B$ and map description $\proj{\cE}$. If
$b = 0$, it outputs the correct result $\cU(\psi_A)$ at Alice's
interface; otherwise it outputs Bob's choice,
$\cE\left(\psi_{AB}\right)$.
\end{deff}

A DQC protocol is verifiable if it provides Alice with a mechanism to
detect a cheating Bob and output an error flag {\tt err} instead of
some incorrect computation. This is modeled by weakening Bob's
filtered functionality: an ideal DQC resource with verifiability,
$\aS^{\blind}_{\verif}$, only allows Bob to input one classical bit $c$, which
specifies whether the output should be $\cU(\psi_A)$ or some error
state $\sket{\err}$, which by construction is orthogonal to the space
of valid outputs. The ideal resource thus never outputs a wrong
computation. This is illustrated in \figref{fig:dqc.ideal.v}.

\begin{deff}
\label{def:dqc.v}
The ideal DQC resource $\aS^{\blind}_{\verif}$ which provides correctness,
blindness and verifiability takes an input $\psi_A$ at Alice's
interface, and two filtered control bits $b$ and $c$ (set by default
to $0$). If $b = 0$, it simply outputs $\cU(\psi_A)$ at Alice's
interface. If $b = 1$, it outputs the permitted leak $\ell^{\psi_A}$ at
Bob's interface, then reads the bit $c$, and conditioned on its value,
it either outputs $\cU(\psi_A)$ or $\sket{\err}$ at Alice's interface.
\end{deff}

\subsubsection{Concrete setting}
\label{sec:dqc.model.concrete}

In the concrete (or real) setting, the only resource that Alice and
Bob need is a (two-way) communication channel $\aR$. Alice's protocol
$\pi_A$ receives $\psi_A$ as an input on its outside interface. It
then communicates through $\aR$ with Bob's protocol $\pi_B$, and
produces some final output $\rho_A$. For the sake of generality we
assume that the operations performed by $\pi_A$ and $\pi_B$, and the
communication between them, are all quantum. Of course, a protocol is
only useful if Alice has very few quantum operations to perform, and
most of the communication is classical. However, to model security, it
is more convenient to consider the most general case possible, so that
it applies to all possible protocols.

As described in \secref{sec:systems.two}, their protocols can be
modeled by a sequence of CPTP maps $\{\cE_i : \lo{AC} \to \lo{AC}
\}_{i=1}^N$ and $\{\cF_i : \lo{CB} \to \lo{CB} \}_{i=1}^{N-1}$. We
illustrate a run of such a protocol in \figref{fig:dqc.real}. This is
a special case of \figref{fig:prot.real} in which Bob has neither
input nor output. The entire system consisting of the protocol
$(\pi_A,\pi_B)$ and the channel $\aR$ is a map which transforms
$\psi_A$ into $\rho_A$. If both players played honestly and the
protocol is correct, this should result in $\rho_A = \cU(\psi_A)$.

\begin{figure}[htb]
\begin{centering}

\begin{tikzpicture}[opnode/.style={minimum width=1.05cm,minimum height=.6cm}]
\small

\def\t{4.1}
\def\u{2.6}
\def\v{.9}

\node[draw,opnode] (a1) at (-\u,0) {$\cE_1$};
\node[draw,opnode] (a2) at (-\u,-\v) {$\cE_2$};
\node[draw,opnode] (a3) at (-\u,-2*\v) {$\cE_3$};
\node[opnode] (a4) at (-\u,-3*\v) {};
\node[opnode] (a45) at (-\u,-3.12*\v) {$\vdots$};
\node[opnode] (a5) at (-\u,-3.5*\v) {};
\node[draw,opnode] (a6) at (-\u,-4.5*\v) {$\cE_N$};
\node[draw,inner sep=.25cm,fit=(a1)(a6)] (a) {};
\node[yshift=-2,above right] at (a.north west) {$\pi_A$};
\node (alice1) at (-\t,0) {};
\node (alice2) at (-\t,-4.5*\v) {};

\node[opnode] (b0) at (\u,0) {};
\node[draw,opnode] (b1) at (\u,-.5*\v) {$\cF_1$};
\node[draw,opnode] (b2) at (\u,-1.5*\v) {$\cF_2$};
\node[opnode] (b3) at (\u,-2.5*\v) {};
\node[opnode] (b35) at (\u,-2.62*\v) {$\vdots$};
\node[opnode] (b4) at (\u,-3*\v) {};
\node[draw,opnode] (b5) at (\u,-4*\v) {$\cF_{N-1}$};
\node[opnode] (b6) at (\u,-4.5*\v) {};
\node[draw,inner sep=.25cm,fit=(b0)(b6)] (b) {};
\node[yshift=-2,above right] at (b.north west) {$\pi_B$};

\node[opnode] (c0) at (0,0) {};
\node[opnode] (c3) at (0,-2.25*\v) {};
\node[opnode] (c4) at (0,-2.85*\v) {$\vdots$};
\node[opnode] (c5) at (0,-3.75*\v) {};
\node[opnode] (c6) at (0,-4.5*\v) {};
\node[draw,minimum width=2.2cm,inner sep=.25cm,fit=(c0)(c6)] (c) {};
\node[yshift=-2,above right] at (c.north west) {$\aR$};

\draw[sArrow] (alice1.center) to node[auto,pos=.4] {$\psi_A$} (a1);
\draw[sArrow] (a1) to (b1);
\draw[sArrow] (a1) to (a2);
\draw[sArrow] (b1) to (a2);
\draw[sArrow] (b1) to (b2);
\draw[sArrow] (a2) to (b2);
\draw[sArrow] (a2) to (a3);
\draw[sArrow] (b2) to (a3);
\draw[sArrow] (b2) to (b3);
\draw[sArrow] (a3) to (c3);
\draw[sArrow] (a3) to (a4);
\draw[sArrow] (c5) to (b5);
\draw[sArrow] (a5) to (a6);
\draw[sArrow] (b4) to (b5);
\draw[sArrow] (b5) to (a6);
\draw[sArrow] (a6) to node[auto,swap,pos=.6] {$\rho_A$} (alice2.center);

\end{tikzpicture}

\end{centering}
\caption[Generic DQC protocol]{\label{fig:dqc.real}A generic run of a
  DQC protocol. Alice has access to the left interface and Bob to the
  right interface. The entire system builds one CPTP operation which
  maps $\psi_A$ to $\rho_A$.}
\end{figure}
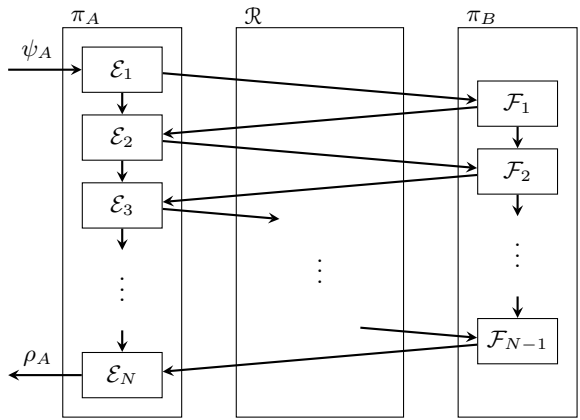

In the following, when we refer to a DQC protocol, we simply mean any
protocol satisfying the model of \figref{fig:dqc.real}. Whether the
protocol actually performs delegated quantum computation depends on
whether it satisfies the correctness condition, which we define in
\secref{sec:dqc.security}.

\subsection{Security of DQC}
\label{sec:dqc.security} 

Since we are interested in modeling a cheating server Bob, but do not
care what happens if the client Alice does not follow her protocol, it
is sufficient to take from \defref{def:cc} the equations corresponding
to an honest Alice. Applying this to the DQC model from the previous
section, we get that a protocol $\pi$ constructs a blind quantum
computation resource $\aS^{\blind}$ from a communication channel $\aR$ within
$\eps$ if there exists a simulator $\sigma_B$ such that
\begin{equation} \label{eq:dqc.b}
  \pi_A\aR\pi_B \close{\eps} \aS^{\blind}\bot_B \qquad \text{and} \qquad
  \pi_A\aR \close{\eps} \aS^{\blind}\sigma_B,
\end{equation}
where $\bot_B$ is a filter which obstructs Bob's cheating
interface.\footnote{From now on, we write all the converters plugged
  in the $A$\=/interfaces on the left of the resources and those
  plugged in the $B$\=/interfaces on the right.} The fist condition in
\eqnref{eq:dqc.b} captures the correctness of the protocol, and we say
that a protocol provides \emph{$\eps$\=/correctness} if this condition
is fulfilled. The second condition, which we illustrate in
\figref{fig:dqc.sec}, measures the security. If it is fulfilled, we
have \emph{$\eps$\=/blindness}. If $\eps = 0$ we say that we have
\emph{perfect blindness}.

\begin{figure}[htb]
\begin{centering}

\begin{tikzpicture}[
      genResource/.style={draw,minimum width=2cm,minimum height=2.4cm},
      innGenResource/.style={minimum width=1.4cm,minimum height=.2cm},
      genConverter/.style={draw,minimum width=1cm,minimum height=2.4cm},
      innGenConverter/.style={minimum width=.4cm,minimum height=.2cm}]

\small

\def\t{-2.5} 
\def\u{.75} 
\def\v{-1.25} 
\def\w{2.5} 
\def\x{.3} 
\def\sh{6.8cm} 

\node[xshift=.52*\sh] at (0,0) {\Large $\close{\eps}$};

\node (in) at (\t,3*\x) {};
\node (out) at (\t,-3*\x) {};

\node[genConverter] (prot) at (\v,0) {};
\node[yshift=-1.5,above right] at (prot.north west) {$\pi_A$};
\node[innGenConverter] (c1) at (\v,3*\x) {};
\node[innGenConverter] (c3) at (\v,\x) {};
\node[innGenConverter] (c4) at (\v,-\x) {};
\node[innGenConverter] (c6) at (\v,-3*\x) {};

\draw[sArrow] (in.center) to node[auto,pos=.4] {$\psi_A$} (c1);
\draw[sArrow] (c6) to node[auto,pos=.6,swap] {$\rho_A$} (out.center);

\node[genResource] (R) at (\u,0) {};
\node[yshift=-1.5,above right] at (R.north west) {$\aR$};

\node (o1) at (\w,3*\x) {};
\node (o2) at (\w,\x) {};
\node (o3) at (\w,-\x) {};
\node (o4) at (\w,-3*\x) {};

\draw[sArrow] (c1) to (o1.center);
\draw[sArrow] (o2.center) to (c3);
\draw[sArrow] (c4) to (o3.center);
\draw[sArrow] (o4.center) to (c6);


\node[xshift=\sh] (iin) at (-\w,3*\x) {};
\node[xshift=\sh] (iout) at (-\w,-3*\x) {};

\node[xshift=\sh,genResource] (iS) at (-\u,0) {};
\node[yshift=-1.5,above right] at (iS.north west) {$\aS$};
\node[xshift=\sh,innGenResource] (is1) at (-\u,3*\x) {};
\node[xshift=\sh,innGenResource] (is2) at (-\u,-3*\x) {};

\draw[sArrow] (iin.center) to node[auto,pos=.4] {$\psi_A$} (is1);
\draw[sArrow] (is2) to node[auto,pos=.6,swap] {$\rho_A$} (iout.center);

\node[xshift=\sh,genConverter] (isim) at (-\v,0) {};
\node[yshift=-1.5,above right] at (isim.north west) {$\sigma_B$};
\node[xshift=\sh,innGenConverter] (ic1) at (-\v,3*\x) {};
\node[xshift=\sh,innGenConverter] (ic3) at (-\v,\x) {};
\node[xshift=\sh,innGenConverter] (ic4) at (-\v,-\x) {};
\node[xshift=\sh,innGenConverter] (ic6) at (-\v,-3*\x) {};

\draw[gArrow] (is1) to (ic1);
\draw[gArrow] (ic6) to (is2);

\node[xshift=\sh] (io1) at (-\t,3*\x) {};
\node[xshift=\sh] (io2) at (-\t,\x) {};
\node[xshift=\sh] (io3) at (-\t,-\x) {};
\node[xshift=\sh] (io4) at (-\t,-3*\x) {};

\draw[sArrow] (ic1) to (io1.center);
\draw[sArrow] (io2.center) to (ic3);
\draw[sArrow] (ic4) to (io3.center);
\draw[sArrow] (io4.center) to (ic6);

\end{tikzpicture}

\end{centering}
\caption[$\eps$-security of DQC]{\label{fig:dqc.sec}An illustration of
  the second terms of \eqnsref{eq:dqc.b} and \eqref{eq:dqc.v}. If a
  distinguisher cannot guess with advantage greater than $\eps$
  whether it is interacting with the real construct on the left or the
  ideal construct on the right, the two are $\eps$\=/close and the
  protocol $\eps$\=/secure against a cheating Bob.}
\end{figure}
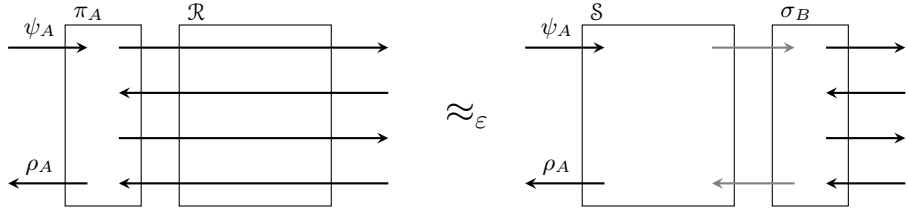

Likewise in the case of verifiability, the ideal resource $\aS^{\blind}_{\verif}$ is
constructed by $\pi$ from $\aR$ if there exists a simulator $\sigma_B$
such that,
\begin{equation}  \label{eq:dqc.v}
  \pi_A\aR\pi_B \close{\eps} \aS^{\blind}_{\verif}\bot_B \qquad \text{and} \qquad
  \pi_A\aR \close{\eps} \aS^{\blind}_{\verif}\sigma_B.
\end{equation}
The first condition from \eqnref{eq:dqc.v} is identical to the first
condition of \eqnref{eq:dqc.b}, and captures $\eps$\=/correctness. The
second condition in \eqnref{eq:dqc.v} (also illustrated by
\figref{fig:dqc.sec}) guarantees both blindness and verifiability, and
if it is satisfied we say that the we have
\emph{$\eps$\=/blind\-/verifiability}.

Note that the exact metrics used to distinguish between the resources
from \eqnsref{eq:dqc.b} and \eqref{eq:dqc.v} are defined in
\secref{sec:systems.dist}. $\pi_A\aR\pi_B$ and $\aS\bot_B$ \--- as can
be seen from their depictions in Figures~\ref{fig:dqc.real} and
\ref{fig:dqc.ideal} (with a filter blocking the cheating interface of
the latter) \--- are resources which implement a single map, so the
diamond distance corresponds to the distinguishing
advantage. $\pi_A\aR$ and $\aS\sigma_B$ are half of two-party
protocols, so the distinguishing metric corresponds to the distance
between quantum strategies/combs introduced by Gutoski and
Watrous~\cite{GW07,Gut12} and Chiribella et al.~\cite{CDP09}, and
described in \secref{sec:systems.dist}.

\section{Blind and verifiable DQC}
\label{sec:verif}

Finding a simulator to prove the security of a protocol can be
challenging. In this section we reduce the task of proving that a DQC
protocol constructs the ideal resource $\aS^{\blind}_{\verif}$ to proving that
the map implemented by the protocol is close to some ideal map that
intuitively provides some form of
local\-/blindness\-/and\-/verifiability. The converse also holds: any
protocol which constructs $\aS^{\blind}_{\verif}$ must be close to this ideal
map.

A malicious server Bob will not apply the CPTP maps assigned to him by
the protocol, but his own set of cheating maps $\{\cF_i : \lo{CB} \to
\lo{CB}\}_{i=1}^{N-1}$. Furthermore, he might hold (the $B$ part of) a
purification of Alice's input, $\psi_{ABR}$. Intuitively, a protocol
provides local\-/blindness\footnoteremember{fn:sa}{We provide formal
  definitions of local\-/blindness and local\-/verifiability in
  \secref{sec:sa.defs}.}  if the final state held by Bob could have
been generated by a local map on his system \--- say, $\cF$ \---
independently from Alice's input, but which naturally depends on his
behavior given by the maps $\{\cF_i \}_i$. It provides
local\-/verifiability\footnoterecall{fn:sa} if the final state held by
Alice is either the correct outcome or some error flag. Combining the
two gives an ideal map of the from $\cU \tensor \cF^{\ok} + \cE^{\err}
\tensor \cF^{\err}$, where $\cF^{\ok}$ and $\cF^{\err}$ break $\cF$
down in two maps which result in the correct outcome and an error
flag, respectively.

\begin{deff}[local-blind-verifiability]
\label{def:verif.bv}
We say that a DQC protocol provides
$\eps$\=/local\-/blind\-/verifiability, if, for all adversarial
behaviors $\{\cF_i\}_i$, there exist two completely positive, trace
non\-/increasing maps $\cF_B^{\ok}$ and $\cF_B^{\err}$, such that
\begin{equation} \label{eq:verif.bv} \cP_{AB} \close{\eps} \cU_A
  \tensor \cF_B^{\ok} + \cE_A^{\err} \tensor
  \cF_B^{\err},\end{equation} where $\cP_{AB} : \lo{AB} \to \lo{AB}$
is the map corresponding to a protocol run with Alice behaving
honestly and Bob using his cheating operations $\{\cF_i\}_i$, and
$\cE^{\err}_A$ discards the $A$ system and produces an error flag
$\proj{\err}$ orthogonal to all possible valid outputs. We say that
the protocol provides $\eps$\=/local\-/blind\-/verifiability for a
set of initial states $\cB$, if \eqnref{eq:verif.bv} holds when
applied to these states, i.e., for all $\psi_{ABR} \in \cB$,
\begin{equation*} \cP_{AB} \left(\psi_{ABR}\right) \close{\eps}
  \left(\cU_A \tensor \cF_B^{\ok} + \cE_A^{\err} \tensor
    \cF_B^{\err}\right) \left(\psi_{ABR}\right).\end{equation*}
\end{deff}

\begin{rem}
  \label{rem:preprocessing}
  For simplicity, this definition assumes the allowed leaks (e.g.,
  input size, computation size) to be fixed, and applies to all
  protocols $\cP_{AB}$ tailored for inputs with an identical leak
  (e.g., identical size). These leaks could be explicitly modeled by
  allowing the maps $\cF_B^{\ok}$ and $\cF_B^{\err}$ to depend on
  them.
\end{rem}

We now prove that it is both necessary and sufficient for a DQC
protocol to satisfy \defref{def:verif.bv} to be blind\-/verifiable,
i.e., to satisfy the second condition of \eqnref{eq:dqc.v}. In order
to construct $\aS^{\blind}_{\verif}$, a DQC protocol also needs to be
$\eps$\=/correct, that is, satisfy the first condition from
\eqnref{eq:dqc.v}. We show in \appendixref{app:cor} that this is
fulfilled, if, when Bob behaves honestly, \eqnref{eq:verif.bv} is
satisfied for $\cF^{\ok}_B = \id_B$ and $\cF^{\err}_B = 0$.

\begin{thm}
\label{thm:verif.sec}
Any DQC protocol which provides $\eps$\=/local-blind\-/verifiability
is $2\eps$\=/blind\-/verifiable. And any DQC protocol which is
$\eps$\=/blind\-/verifiable provides
$\eps$\=/local\-/blind\-/verifiability.
\end{thm}

To show that local\-/blind\-/verifiability implies
blind\-/verifiability we use a standard ``dummy input'' argument: the
simulator runs Alice's protocol with a dummy input, and notifies the
ideal resource to abort if the simulation aborts. The converse is
immediate after writing up the combined actions of the distinguisher
and simulator as maps.

\begin{proof}
  We start by showing that local\-/blind\-/verifiability is sufficient
  for a DQC protocol to be blind\-/verifiable, i.e., there exists a
  simulator $\sigma_B$ such that the two resources in
  \figref{fig:dqc.sec} are $2\eps$\=/close. To do this, we define
  $\sigma_B$ to work as follows. It sets the bit $b = 1$, receives the
  permitted leaks $\ell^{\psi_A}$ from $\aS^{\blind}_{\verif}$, picks
  any input $\psi_{B}$ compatible with this information, and runs the
  protocol $\pi_A$ on this input with its internal register, which we
  denote by $B$. After the last step, it projects the state it holds
  in $B$ on $\proj{\err}$ and $I-\proj{\err}$, and sends $c=0$ to
  $\aS^{\blind}_{\verif}$ if no error was detected, otherwise it sends
  $c=1$. As defined in \defref{def:dqc.v}, $\aS^{\blind}_{\verif}$
  then either outputs the correct result or an error flag depending on
  the value of $c$.

  As described in \secref{sec:systems.dist}, the most general
  operation the distinguisher can perform to distinguish between the
  resources $\pi_A\aR$ and $\aS^{\blind}_{\verif}\sigma_B$, is to choose some initial
  state $\psi_{AR}$, send $\psi_A$ to the system with which it is
  interacting, apply some operations $\{\cF_i : \lo{CR} \to \lo{CR}
  \}_{i=1}^{N-1}$ each time it receives some message on the channel
  $C$, and return each time the new state in $C$.

  Let $\rho^{\psi}_{AR}$ be the final state when the distinguisher is
  interacting with $\pi_A\aR$. By \eqnref{eq:verif.bv},\footnote{In
    the real system, Alice (holding $A$) runs the protocol with the
    distinguisher (holding $R$). With these indices
    \eqnref{eq:verif.bv} reads $\cP_{AR} \close{\eps} \cU_A \tensor
    \cF_R^{\ok} + \cE_A^{\err} \tensor \cF_R^{\err}$.} this state is
  $\eps$-close to
  \[ \tau^{\psi}_{AR} \coloneqq \left( \cU \tensor \cF^{\ok} \right)
  (\psi_{AR}) + \proj{\err} \tensor \cF^{\err} (\psi_{R}),\] for some
  $\cF^{\ok}$ and $\cF^{\err}$ which depend only on $\{\cF_i\}_i$, not
  on $\psi_{AR}$.

  When the distinguisher is interacting with $\aS\sigma_B$ and using
  the same operations $\{\cF_i\}_i$ and initial state $\psi_{AR}$, let
  $\alpha^{\psi}_{ARB}$ be the state of the system at the end of the
  subroutine $\pi_A$ and before sending the bit $c$ to
  $\aS^{\blind}_{\verif}$. Then, using
  \eqnref{eq:verif.bv},\footnote{In the ideal system, the simulator
    (holding $B$) runs the protocol with the distinguisher (holding
    $R$). With these indices \eqnref{eq:verif.bv} reads $\cP_{BR}
    \close{\eps} \cU_B \tensor \cF_R^{\ok} + \cE_B^{\err} \tensor
    \cF_R^{\err}$.} we find that $\alpha^{\psi}_{ARB}$ is $\eps$-close
  to
  \[ \gamma^{\psi}_{ARB} \coloneqq \left( \id_A \tensor \cF^{\ok}
    \tensor \cU \right) (\psi_{AR} \tensor \psi_{B}) + \left(
    \id_A \tensor \cF^{\err}\right) (\psi_{AR}) \tensor \proj{\err}.\]

  The final operation performed by $\aS^{\blind}_{\verif}$ to generate the output can be
  seen as a map $\cS$, which conditioned on $B$ being an error,
  deletes $B$ and overwrites $A$ with an error, and conditioned
  on $B$ being a valid output, deletes $B$ and applies
  $\cU$ to the system $A$. Since a map can only decrease the distance
  between two states, the final state of the system after this
  operation, $\phi^{\psi}_{AR} \coloneqq
  \cS(\alpha^{\psi}_{ARB})$, is $\eps$-close to
  $\cS(\gamma^{\psi}_{ARB}) = \tau^{\psi}_{AR}$. By the triangle
  inequality we thus have $\rho^{\psi}_{AR} \close{2\eps}
  \phi^{\psi}_{AR}$.

  We now prove the converse. If the protocol is
  $\eps$\=/blind\-/verifiable, there exists a simulator $\sigma_B$
  such that $\pi_A\aR \close{\eps}\aS\sigma_B$. A distinguisher
  interacting with one of the two systems chooses an initial state
  $\psi_{AR}$, and applies operations $\cF_i : \lo{CR} \to \lo{CR}$ to
  the messages received on the channel $C$ and the system $R$.

  Consider now the interaction of the simulator and the
  distinguisher. Since the simulator deletes its internal memory when
  it terminates, and outputs only a single bit $c$ notifying the ideal
  resource to output the correct result or an error flag, the combined
  action of the two can be seen as a CPTP map $\cF : \lo{R} \to
  \{0,1\} \times \lo{R}$. Conditioning on the output $\{0,1\}$, we
  explicitly define two trace non\-/increasing maps
  $\cF^{\ok},\cF^{\err} : \lo{R} \to \lo{R}$, i.e., $\cF(\rho) =
  \proj{0} \tensor \cF^{\ok}(\rho) + \proj{1} \tensor
  \cF^{\err}(\rho)$. Since the ideal blind and verifiable DQC resource
  outputs the correct result upon receiving $0$, and an error flag
  otherwise, the joint map of ideal resource, simulator and
  distinguisher is given by $\cU \tensor \cF^{\ok} + \cE^{\err}
  \tensor \cF^{\err}$. And this map must be $\eps$-close to the real
  map, otherwise the distinguisher would have an advantage greater
  than $\eps$.
\end{proof}

\section{Reduction to local criteria}
\label{sec:standalone}

Although the notion of local\-/blind\-/verifiability defined in the
previous section captures the security of DQC in a single equation, it
is still more elaborate than existing definitions found in the
literature, that treat blindness and verifiability separately.

In \secref{sec:sa.defs} we provide separate definitions for these
local notions, and strengthen local\-/verifiability by requiring that
the server Bob be able to infer on his own whether the client Alice
will reject his response \--- learning whether Alice did reject will
then not provide him with any information that he could not obtain on
his own. In \secref{sec:sa.comp} we show that in the case where Bob
does not hold a state entangled with the input (e.g., when the input is
entirely classical), these notions are sufficient to obtain
local\-/blind\-/verifiability with a similar error parameter. In the
case where Bob's system is entangled to Alice's input, we show that
the same holds, albeit with an error increased by a factor $\left(\dim
  \hilbert_{A_Q}\right)^2$, where $A_Q$ is the subsystem of Alice's
input which is quantum.

This can be used to show that the protocol of Fitzsimons and
Kashefi~\cite{FK12} and Morimae~\cite{Mor14}, which have already been
analyzed using (insufficient) local criteria, are secure. We provide a
proof sketch of the missing steps for both these protocols
in \appendixref{app:app}.

\subsection{Local-blindness and independent local-verifiability}
\label{sec:sa.defs}

Local\-/blindness can be seen as a simplification of local\-/blind\-/
verifiability, in which we ignore Alice's outcome and only
check that Bob's system could have been generated locally, i.e., is
independent from Alice's input (and output).

\begin{deff}[Local-blindness]
\label{def:sa.b}
A DQC protocol provides $\eps$\=/local\-/blindness, if, for all
adversarial behaviors $\{\cF_i\}_i$, there exists a CPTP map $\cF :
\lo{B} \to \lo{B}$ such that 
\begin{equation} \label{eq:sa.b} \tr_A \circ \cP_{AB} \close{\eps} \cF
  \circ \tr_A,\end{equation} where $\circ$ is the composition of maps,
$\tr_A$ the operator that trace out the $A$\=/system, and $\cP_{AB} :
\lo{AB} \to \lo{AB}$ is the map corresponding to a protocol run with
Alice behaving honestly and Bob using his cheating operations
$\{\cF_i\}_i$. We say that the protocol provides
$\eps$\=/local\-/blindness for a set of initial states $\cB$, if
\eqnref{eq:sa.b} holds when applied to these states, i.e., for all
$\psi_{ABR} \in \cB$, \[ \tr_A \circ \cP_{AB} (\psi_{ABR})\close{\eps}
\cF \circ \tr_A (\psi_{ABR}).\]
\end{deff}

Likewise, local\-/verifiability can also be seen as a simplification
of local\-/blind\-/verifiability, in which we ignore Bob's system and
only check that Alice holds either the correct outcome or an error
flag $\sket{\err}$, which by construction is orthogonal to any
possible valid output. In the following we define
local\-/verifiability only for the case where Bob's system is not
entangled to Alice's input, since otherwise the correct outcome
depends on Bob's actions, and cannot be modeled by describing Alice's
system alone.\footnote{The resulting definition is equivalent to that
  of \cite{FK12} and non\-/composable authentication
  definitions~\cite{BCGST02}, which bound the probability of
  projecting the outcome on the space of invalid results.}


\begin{deff}[Local\-/verifiability] \label{def:sa.v} A DQC
  protocol provides $\eps$\=/local\-/verifiability, if, for all
  adversarial behaviors $\{\cF_i\}_i$ and all initial states $\psi_{AR_1}
  \tensor \psi_{R_2B}$, there exists a $0 \leq p^\psi \leq 1$ such
  that
  \begin{equation} \label{eq:sa.v.td} \rho^{\psi}_{AR_1} \close{\eps}
    p^\psi \left(\cU \tensor \id_{R_1}\right)(\psi_{AR_1}) + (1 -
    p^\psi) \proj{\err} \tensor \psi_{R_1}, \end{equation} where
  $\rho^{\psi}_{AR_1}$ is the final state of Alice and the first part
  of the reference system. We say that the protocol provides
  $\eps$\=/local\-/verifiability for a set $\cB$ of initial states in
  product form, if \eqnref{eq:sa.v.td} holds for all $\psi_{AR_1}
  \tensor \psi_{R_2B} \in \cB$.
\end{deff}

As mentioned in \secref{sec:intro}, local\-/blindness and
local\-/verifiability together do not provide the security guarantees
one expects from DQC. This seems to be because the verification
procedure can depend on the input (as in the example from
\footnoteref{fn:gap}), and thus if Bob learns the result of this
measurement, he learns something about the input. This motivates us to
define a stronger notion, in which Bob can reconstruct on his own
whether the output will be accepted \--- the outcome of Alice's
verification procedure must thus be independent of her input. To do
this, we introduce a new qubit in a system $\bar{B}$, which contains a
copy of the information whether Alice accepts or rejects, i.e., for a
final state
\begin{equation}
\label{eq:sa.ind.final}
\rho^{\psi}_{ARB} = \phi^{\ok}_{ARB} + \proj{\err}
\tensor \phi^{\err}_{RB},
\end{equation}
we define
\begin{equation}
\label{eq:sa.ind.post} \rho^{\psi}_{ARB\bar{B}} \coloneqq 
\phi^{\ok}_{ARB} \tensor \proj{\ok} + \proj{\err} 
\tensor \phi^{\err}_{RB} \tensor \proj{\err}.
\end{equation}
Note that \eqnref{eq:sa.ind.post} can be generated from
\eqnref{eq:sa.ind.final} by introducing a system $\bar{B}$ in the
state $\ket{\ok}$ and changing its value to $\ket{\err}$ conditioned
on $A$ being in the state $\ket{\err}$. Let $\cQ_{A\bar{B}} : \lo{A}
\to \lo{A\bar{B}}$ be such an operation, i.e.,
$\rho^{\psi}_{ARB\bar{B}} =
\cQ_{A\bar{B}}(\rho^{\psi}_{ARB})$. \eqnref{eq:sa.ind.final} can then
be recovered from \eqnref{eq:sa.ind.post} by tracing out the system
$\bar{B}$.

The notion of verifiability is strengthened by
additionally requiring that leaking this system $\bar{B}$ to the
adversary does not provide him with more information about the input,
i.e., Bob could (using alternative maps) generate the system $\bar{B}$
on his own.

\begin{deff}
\label{def:sa.ind}
A DQC protocol provides $\bar{\eps}$\=/independent
$\eps$\=/local\-/verifiability, if, in addition to providing
$\eps$\=/local\-/verifiability, for all adversarial behaviors $\{\cF_i
: \lo{CB} \to \lo{CB}\}_i$ there exist alternative maps $\{\cF'_i :
\lo{CB\bar{B}} \to \lo{CB\bar{B}}\}_i$ (for an initially empty system
$\bar{B}$), such that
\begin{equation} \label{eq:sa.ind} \tr_A \circ \cQ_{A\bar{B}} \circ
  \cP_{AB} \close{\bar{\eps}} \tr_{A} \circ
  \cP'_{AB\bar{B}},\end{equation} where $\circ$ is the composition of
maps, $\cP_{AB} : \lo{AB} \to \lo{AB}$ and $\cP'_{AB\bar{B}} : \lo{AB}
\to \lo{AB\bar{B}}$ are the maps corresponding to runs of the protocol
with Alice being honest and Bob using maps $\{\cF_i\}_i$ and
$\{\cF'_i\}_i$ respectively, and $\cQ_{A\bar{B}} : \lo{A} \to
\lo{A\bar{B}}$ is a map which generates from $A$ a system $\bar{B}$
holding a copy of the information whether Alice accepts or rejects. We
say that a protocol provides $\bar{\eps}$\=/independent
$\eps$\=/local\-/verifiability for a set of initial states $\cB$, if
the same conditions hold for all states in $\cB$, i.e., if we have
$\eps$\=/local\-/verifiability for $\cB$, and if for all $\psi_{ABR}
\in \cB$, \[ \tr_A \circ \cQ_{A\bar{B}} \circ \cP_{AB} (\psi_{ABR})
\close{\bar{\eps}} \tr_{A} \circ \cP'_{AB\bar{B}} (\psi_{ABR}).\]
\end{deff}

\begin{rem}
  \label{rem:sa.ind} By the triangle inequality, if a protocol
  provides both $\eps$\=/local\-/blindness and
  $\bar{\eps}$\=/independent $\eps'$\=/local\-/verifiability, then
  there exists a map $\cF' : \lo{B} \to \lo{B\bar{B}}$ such that
  \begin{equation}
    \label{eq:sa.ind.rem}
    \tr_A \circ \cQ_{A\bar{B}} \circ \cP_{AB} \close{\eps+\bar{\eps}} \cF' \circ \tr_A.
  \end{equation}
\end{rem}

\subsection{Reduction}
\label{sec:sa.comp}

We first show in \lemref{lem:sa.comp} that in the special case of
initial states which are not entangled between Alice and Bob's systems
(e.g., the input is classical), local\-/blindess and independent
local\-/verifiability are sufficient to achieve
local\-/blind\-/verifiability. In \thmref{thm:sa.comp} we then
generalize this to any initial state.

\begin{rem} \label{rem:sa.comp} The two proofs in this section only
  hold for protocols that construct a DQC resource for which the
  implemented operation $\cU$ is unitary. Since any quantum operation
  can be written as a unitary on a larger system~\cite{NC00}, this
  effectively allows the theorems to apply to any CPTP operation $\cE$
  as long as the necessary qubits for the unitary implementation are
  appended to the in- and outputs. For example, instead of defining
  universal computation as a unitary, most papers \--- e.g.,
  \cite{BFK09,FK12,MF13,Mor14} \--- describe how to perform any
  (arbitrary) unitary operation $U_x$ on any arbitrary input
  $\rho_{\text{in}}$. By appending the description $x$ of the unitary
  $U_x$ to the input and output, this is equivalent to applying the
  unitary transformation $\cU \coloneqq \sum_x U_x \tensor \proj{x}$
  to the input $\rho_{\text{in}} \tensor \proj{x}$.
\end{rem}

\begin{lem}\label{lem:sa.comp}  If a DQC protocol implementing a
  unitary transformation provides
  $\eps_{\text{bl}}$\=/local\-/blindness and
  $\eps_{\text{ind}}$\=/independent
  $\eps_{\text{ver}}$\=/local\-/verifiability for any pure initial
  state of the form $\psi_{AR_1} \tensor \psi_{R_2B}$, then the
  protocol provides $\delta$\=/local\-/blind\-/verifiability with
  $\delta = 2\sqrt{2\eps}_{\text{ver}} + \eps_{\text{bl}} +
  \eps_{\text{ind}}$ for these initial states in product form.
\end{lem}

\begin{proof}
In this proof, we use several times the following simple equality. For
two states $\rho = \proj{0} \tensor \rho_0 + \proj{1} \tensor \rho_1$
and $\sigma = \proj{0} \tensor \sigma_0 + \proj{1} \tensor \sigma_1$,
we have \begin{equation} \label{eq:sa.lem.tool}
D(\rho,\sigma) = D(\rho_0,\sigma_0) + D(\rho_1,\sigma_1).
\end{equation}

In \remref{rem:sa.ind} we combined the conditions of local\-/blindness
and the new condition of independent local\-/verifiability into one
new formula, \eqnref{eq:sa.ind.rem}. It is thus sufficient to prove
that if \eqnref{eq:sa.ind.rem} and \eqnref{eq:sa.v.td}, are satisfied
for any pure product initial state $\psi_{AR_1} \tensor \psi_{R_2B}$,
then we have local\-/blind\-/verifiability,
i.e., \begin{multline} \label{eq:sa.lem.bv} \rho^\psi_{AR_1R_2B}
  \close{\delta} \left(\cU \tensor \id_{R_1R_2} \tensor \cF^{\ok}
  \right) (\psi_{AR_1} \tensor \psi_{R_2B}) \\ + \proj{\err} \tensor
  \psi_{R_1} \tensor \left(\id_{R_2} \tensor \cF^{\err}
  \right)(\psi_{R_2B}),\end{multline} for some $\cF^{\ok}$ and
$\cF^{\err}$.

Since $\sket{\err}$ is orthogonal to any valid output, both the RHS of
\eqnref{eq:sa.lem.bv} and LHS (given in \eqnref{eq:sa.ind.final}) are
a linear combination of orthogonal states on the same subspaces. And
thus by \eqnref{eq:sa.lem.tool}, to show that \eqnref{eq:sa.lem.bv}
holds for some $\delta$, it is sufficient to find maps $\cF^{\ok}$ and
$\cF^{\err}$, and $\delta_1$ and $\delta_2$ with $\delta_1 + \delta_2
= \delta$, such that
\begin{align}
  \phi^{\ok}_{AR_1R_2B} & \close{\delta_1} \left(\cU \tensor \id_{R_1R_2} \tensor
    \cF^{\ok} \right) (\psi_{AR_1} \tensor \psi_{R_2B}), \label{eq:sa.lem.goal.1} \\
  \phi^{\err}_{R_2B} & \close{\delta_2} \left(\id_{R_2} \tensor \cF^{\err}
  \right)(\psi_{R_2B}). \label{eq:sa.lem.goal.2}
\end{align}

Let $\cF' : \lo{B} \to \lo{B\bar{B}}$ be the map guaranteed to exist
by the combination of local\-/blindness and independent
local\-/verifiability (\eqnref{eq:sa.ind.rem}), and let
$\cP^{\ok}_{\bar{B}}$ and $\cP^{\err}_{\bar{B}}$ be the maps
corresponding to projections on the states $\ket{\ok}$ and
$\ket{\err}$ of the $\bar{B}$ system. We define
\begin{align*}
  \cF^{\ok}_B & \coloneqq \tr_{\bar{B}} \circ \cP^{\ok}_{\bar{B}} \circ \cF',\\
  \cF^{\err}_B & \coloneqq \tr_{\bar{B}} \circ \cP^{\err}_{\bar{B}} \circ \cF'.
\end{align*}
Note that w.l.o.g., we can take $\cF'$ to generate a linear
combination of two orthogonal states, one in the $\cP^{\ok}_{\bar{B}}$
subspace and one in the $\cP^{\err}_{\bar{B}}$. Thus, applying
\eqnref{eq:sa.ind} to the initial state $\psi_{AR_1} \tensor
\psi_{R_2B}$ and using \eqnref{eq:sa.lem.tool}, we find that there
exist $\eps_1$ and $\eps_2$ with $\eps_1 + \eps_2 = \eps_{\text{ind}}
+ \eps_{\text{bl}}$ such that
\begin{align}
  \phi^{\ok}_{R_2B} & \close{\eps_1} \left(\id_{R_2} \tensor \cF^{\ok}_B\right)(\psi_{R_2B}), \label{eq:sa.lem.blind.1} \\
  \phi^{\err}_{R_2B} & \close{\eps_2} \left(\id_{R_2} \tensor
    \cF^{\err}_B\right) (\psi_{R_2B}). \label{eq:sa.lem.blind.2}
\end{align}
Note that \eqnref{eq:sa.lem.blind.2} is exactly one of the conditions
we need to find, namely \eqnref{eq:sa.lem.goal.2}. We now still need
to bound \eqnref{eq:sa.lem.goal.1}.

We take the definition of local\-/verifiability, \eqnref{eq:sa.v.td}; again,
both the RHS and LHS (defined in \eqnref{eq:sa.ind.final}) are linear
combinations of orthogonal states on the same subspaces, hence there
exist $\bar{\eps}_1$ and $\bar{\eps}_2$ with $\bar{\eps}_1 +
\bar{\eps}_2 = \eps_{\text{ver}}$, such that
\begin{align}
  \phi^{\ok}_{AR_1} & \close{\bar{\eps}_1} p^{\psi} \left(\cU \tensor
    \id_{R_1} \right) (\psi_{AR_1}), \label{eq:sa.lem.verif.1} \\
  \trace{\phi^{\err}_{R_2B}} & \close{\bar{\eps}_2}
  1-p^{\psi}. \label{eq:sa.lem.verif.2}
\end{align}

From \eqnref{eq:sa.lem.verif.2} we have that
$\trace{\phi^{\ok}_{AR_1}} = 1 - \trace{\phi^{\err}_{R_2B}}
\close{\bar{\eps}_2} p^{\psi}$. The generalized trace distance
(see \appendixref{app:dist}) between the two states from
\eqnref{eq:sa.lem.verif.1} is thus bounded by
$\bar{D}\left(\phi^{\ok}_{AR_1}, p^{\psi}\cU(\psi_{AR_1})\right) \leq
\bar{\eps}_1 + \bar{\eps}_2 = \eps_{\text{ver}}$. From
\lemref{lem:purified}, we can upper bound the purified distance with
the generalized trace distance, and get $P\left(\phi^{\ok}_{AR_1},
  p^{\psi}\cU(\psi_{AR_1})\right) \leq \sqrt{2\eps_{\text{ver}}}$. We
can now apply Uhlmann's theorem to the purified distance (see
\lemref{lem:uhlmann}) and find that since $\cU(\psi_{AR_1})$ is a pure
state, there exists a $\sigma_{R_2B}$ such that
$P\left(\phi^{\ok}_{AR_1R_2B}, p^{\psi}\cU(\psi_{AR_1}) \tensor
  \sigma_{R_2B}\right) = P\left(\phi^{\ok}_{AR_1},
  p^{\psi}\cU(\psi_{AR_1})\right)$. Hence by \lemref{lem:purified},
\eqnref{eq:sa.lem.blind.1}, and the triangle inequality,
\begin{align*}
  & D\left(\phi^{\ok}_{AR_1R_2B},\cU(\psi_{AR_1}) \tensor
    \cF^{\ok}_B(\psi_{R_2B}) \right) \\
  & \qquad \qquad \leq
  D\left(\phi^{\ok}_{AR_1R_2B},p^{\psi}\cU(\psi_{AR_1}) \tensor
    \sigma_{R_2B}\right) \\ & \qquad \qquad \qquad \qquad +
  D\left(p^{\psi}\cU(\psi_{AR_1}) \tensor \sigma_{R_2B},
    \cU(\psi_{AR_1}) \tensor
    \cF^{\ok}_B(\psi_{R_2B})\right) \\
  & \qquad\qquad \leq \sqrt{2\eps_{\text{ver}}} +
  D\left(p^{\psi}\sigma_{R_2B},\phi^{\ok}_{R_2B}\right)
  + D\left(\phi^{\ok}_{R_2B},\cF^{\ok}_B(\psi_{R_2B})  \right) \\
  & \qquad \qquad \leq 2\sqrt{2\eps_{\text{ver}}} + \eps_1.
\end{align*}
Combining this with our bound for \eqnref{eq:sa.lem.goal.2}, we
prove the lemma.
\end{proof}

We now generalize this lemma to initial states that may be entangled
between Alice and Bob. Since protocols can require part of Alice's
input to be classical, we consider initial states of the form
$\psi_{A_CA_QBR}$, where the register $A_C$ is classical, $A_Q$ is
quantum, and $A_QBR$ may be arbitrarily entangled. We reduce this case
to the separable state case treated in \lemref{lem:sa.comp} with an
increase of the error by a factor of $(\dim \hilbert_{A_Q})^2$.

\begin{thm}\label{thm:sa.comp}  If a DQC protocol implementing a
  unitary transformation provides
  $\eps_{\text{bl}}$\=/local\-/blindness and
  $\eps_{\text{ind}}$\=/independent
  $\eps_{\text{ver}}$\=/local\-/verifiability, then it provides
  $\delta$\=/local\-/blind\-/verifiability with $\delta =
  N^2(2\sqrt{2\eps_{\text{ver}}} + \eps_{\text{bl}} +
  \eps_{\text{ind}})$, for $N = \dim \hilbert_{A_Q}$, the dimension of
  the subsystem of Alice's input which is quantum.
\end{thm}

\begin{proof}
  For any initial state $\psi_{ABR} = \proj{x}_{A_C} \tensor
  \psi_{A_QBR}$ and $n \coloneqq \log \dim \hilbert_{A_Q}$, we define
  the state $\psi'_{ATBRS} \coloneqq \proj{x}_{A_C} \tensor
  \proj{\Phi^+}^{\tensor n}_{A_QT} \tensor \psi'_{BRS}$, where
  $\sket{\Phi^+} = (\sket{00}+\sket{11})/\sqrt{2}$ is an EPR pair and
  $\psi'_{BRS} = \psi_{BRA_Q}$. For any map $\cE_{AB} : \lo{AB} \to
  \lo{AB}$ we have \[ \cE_{AB}(\psi_{ABR}) = 2^{2n}
  \trace[TS]{\proj{\Phi^+}^{\tensor n}_{TS}
    \cE_{AB}(\psi'_{ATBRS}) \proj{\Phi^+}^{\tensor n}_{TS}}.\]
  The projection on $\proj{\Phi^+}^{\tensor n}_{TS}$ can be seen as a
  teleportation of the system $S$ into $A_Q$ with a post\-/selection
  on the branch where no bit or phase corrections are necessary.

  Let $\cQ_{AB} : \lo{AB} \to \lo{AB}$ be the map corresponding to a
  run of the protocol with Alice behaving honestly and Bob using his
  cheating strategy. Furthermore, let $\cF' : \lo{B} \to
  \lo{B\bar{B}}$ be the map guaranteed to exist by the combination of
  local\-/blindness and independent local\-/verifiability
  (\eqnref{eq:sa.ind.rem}), and let $\cP^{\ok}_{\bar{B}}$ and
  $\cP^{\err}_{\bar{B}}$ be the maps corresponding to projections on
  the states $\ket{\ok}$ and $\ket{\err}$ of the $\bar{B}$ system. We
  define
  \begin{align*}
    \cF^{\ok}_B & \coloneqq \tr_{\bar{B}} \circ \cP^{\ok}_{\bar{B}} \circ \cF',\\
    \cF^{\err}_B & \coloneqq \tr_{\bar{B}} \circ \cP^{\err}_{\bar{B}} \circ \cF',\\
    \cR_{AB} & \coloneqq \cU \tensor \cF^{\ok}_B + \cE^{\err}_A \tensor \cF^{\err}_B, 
  \end{align*}
  where $\cU$ is the map implemented by the DQC protocol and
  $\cE^{\err}_A$ deletes the contents of $A$ and outputs the error flag
  $\sket{\err}$.

  We then have, \begin{align*}
      & D\left(\cQ_{AB}(\psi_{ABR}),\cR_{AB}(\psi_{ABR})\right) \\
      & \qquad \qquad  =  2^{2n} D\left(\trace[TS]{\proj{\Phi^+}^{\tensor
      n}_{TS} \cQ_{AB}(\psi'_{ATBRS}) \proj{\Phi^+}^{\tensor
      n}_{TS}}, \right. \\
      & \qquad \qquad \qquad  \qquad \qquad  \hfill \left. \trace[TS]{\proj{\Phi^+}^{\tensor
      n}_{TS} \cR_{AB}(\psi'_{ATBRS}) \proj{\Phi^+}^{\tensor
      n}_{TS}} \right) \\
      & \qquad \qquad  \leq 2^{2n} D\left(\cQ_{AB}(\psi'_{ATBRS}),\cR_{AB}(\psi'_{ATBRS})\right).
    \end{align*}
  
    Note that the state $\psi'_{ATBRS}$ is in product form w.r.t.\ the
    systems $AT$ and $BRS$. This allows us to use
    \lemref{lem:sa.comp}, from which we get
    $$D\left(\cQ_{AB}(\psi'_{ATBRS}),\cR_{AB}(\psi'_{ATBRS})\right)
    \leq   2\sqrt{2\eps}_{\text{ver}} + \eps_{\text{bl}} +
    \eps_{\text{ind}}.$$

    By linearity this applies to any initial state $\psi_{A_CA_QBR}$
    classical on $A_C$.
\end{proof}

\begin{rem} \label{rem:failure} If the input is entirely classical
  (e.g., the client wants to factor a number), the failure $\eps$ is
  polynomial in the error parameters of the different local criteria,
  and the reduction is tight. If the input is quantum, the failure is
  multiplied by the dimension squared of the quantum (sub)system, and
  the errors of the local criteria need to be exponentially small in
  the size of the quantum input to compensate. \end{rem}

\begin{cor}[\thmref{thm:main} restated]
\label{cor:sa.comp}
If a DQC protocol implementing a unitary transformation provides
$\eps_{\text{bl}}$\=/local\-/blindness and
$\eps_{\text{ind}}$\=/independent
$\eps_{\text{ver}}$\=/local\-/verifiability for all inputs
$\psi_{A_CA_Q}$, where $A_C$ is classical and $A_Q$ is quantum, then
it is $\delta N^2$\=/blind\-/verifiable, where $\delta =
4\sqrt{2\eps_{\text{ver}}} + 2\eps_{\text{bl}} + 2\eps_{\text{ind}}$
and $N = \dim \hilbert_{A_Q}$. If additionally it provides
$\eps_{\text{cor}}$\=/local\-/correctness,\footnote{See
  \defref{def:cor.sa} on \pref{def:cor.sa}.} it constructs
$\aS^{\blind}_{\verif}$ from a communication channel within $\eps =
\max\{\delta N^2,\eps_{\text{cor}}\}$.
\end{cor}

\begin{proof}
  Immediate by combining \thmref{thm:verif.sec}, \thmref{thm:sa.comp}
  and \lemref{lem:cor}.
\end{proof}

\section{Blindness without verifiability}
\label{sec:blind}

We prove in this section that two different DQC protocols proposed in
the literature construct the ideal blind quantum computation resource
$\aS^{\blind}$ given in \defref{def:dqc.b}. To show this, we need to
prove that both conditions from \eqnref{eq:dqc.b} are satisfied for
$\eps = 0$. In \appendixref{app:cor} we show that the intuitive notion
of local\-/correctness used in the literature is in fact composable,
and thus the first part of \eqnref{eq:dqc.b} is immediate from
existing literature. In the following sections, we prove that these
protocols also provide perfect blindness. Note that they do not
provide verifiability, we therefore cannot use the generic results
from \secref{sec:standalone} to prove that they are blind.

We start in \secref{sec:blind.bfk} with the DQC protocol of Broadbent,
Fitzsimons and Kashefi~\cite{BFK09}, which we describe in detail in
\secref{sec:blind.bfk.protocol}. In this protocol, Alice hides the
computation by encrypting it with a one-time pad. The core idea used
to construct the simulator can also be used to prove the security of
the one-time pad. In \secref{sec:blind.bfk.otp} we thus first sketch
the security proof of the one-time pad, and in
\secref{sec:blind.bfk.security} we prove that the DQC protocol of
Broadbent, Fitzsimons and Kashefi provides perfect blindness.

Morimae and Fujii~\cite{MF13} proposed a DQC protocol with one-way
communication from Bob to Alice, in which Alice simply measures each
qubit she receives, one at a time. We show in
\secref{sec:blind.oneway} that the general class of protocols with
one-way communication is perfectly blind.

\subsection{DQC protocol of Broadbent, Fitzsimons and Kashefi}
\label{sec:blind.bfk}

\subsubsection{The protocol}
\label{sec:blind.bfk.protocol}

This protocol~\cite{BFK09} was originally called Universal Blind
Quantum Computation (UBQC), and in the following we use this name. For
an overview of the UBQC protocol, we assume familiarity with
measurement\-/based quantum computing, for more details see
\cite{RB01,DKP07}. Suppose Alice has in mind a unitary operator $U$
that is implemented with a measurement pattern on a brickwork state
$\mathcal{G}_{n\times (m+1)}$ (\figref{fig:cluster}) with measurements
given as multiples of~$\pi/4$ in the $(X,Y)$ plane with overall
computation size $S = n \times (m+1).$ Note that measurement based
quantum computation, where the measurements are restricted in the
sense above is approximately universal, so there are no restrictions
imposed on $U$ \cite{BFK09}.

 \begin{figure}[htb] \centering
\includegraphics[scale=0.8]{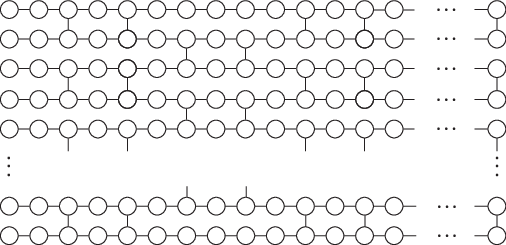}

\caption[Brickwork state]{\label{fig:cluster}The \emph{brickwork
    state}, $\mathcal{G}_{n \times m}$, a universal resource state for
  measurement-based quantum computing requiring only single qubit
  measurement in the $(X,Y)$ plane \cite{BFK09}.  Qubits
  $\sket{\psi_{x,y}}$ $(x=1, \ldots, n, y=1,\ldots,m)$ are arranged
  according to layer $x$ and row $y$, corresponding to the vertices in
  the above graph, and are originally in the~$\sket{+} =
  \frac{1}{\sqrt{2}}\left(\sket{0} +\sket{1}\right)$ state.
  Controlled-$Z$ gates are then performed between qubits which are
  joined by an edge.  The rule determining which qubits are joined by
  an edge is as follows: 1) Neighboring qubits of the same row are
  joined; 2) For each column $j = 3 \mod\, 8$ and each odd row $i$,
  the qubits at positions $(i, j)$ and $(i + 1, j)$ and also on
  positions $(i, j + 2)$ and $(i + 1, j + 2)$ are joined; 3) For each
  column $j = 7 \mod\, 8$ and each even row $i$, the qubits at
  positions $(i, j)$ and $(i + 1, j)$ and also on positions $(i, j +
  2)$ and $(i + 1, j + 2)$ are joined.  The quantum input is usually
  placed in the leftmost column of the brickwork state, whereas the
  output is generated in the rightmost column by sequential single
  qubit measurements. The qubits are usually measured from top to
  bottom per column, where the order of columns is from left to
  right.}
 \end{figure}

 This pattern could have been designed either directly in MBQC or
 generated from a circuit construction.  Each qubit in $
 \mathcal{G}_{n \times (m+1)}$ is indexed by a \emph{column} $y \in
 \{0, \ldots ,m\}$ and \emph{row}~$x \in \{1, \ldots ,n\} = [n]$. Thus
 each qubit is assigned a measurement angle~$\phi_{x,y}$, and two sets
 $D_{x,y},D'_{x,y} \subseteq [n] \times \{0,\ldots, y-1\}$ which we call
 $X$\=/dependencies and $Z$\=/dependencies, respectively.

The dependency sets comprise subsets of the set of the two-coordinate
indices. They reflect the fact that in measurement-based quantum
computation, to ensure a correct and deterministic computation, the
measurement angles which define the computation may have to be
modified for each qubit depending on some of the prior measurement
outcomes.  In particular, here we assume that the dependency
sets~$D_{x,y}$ and~$D_{x,y}'$ are obtained via the flow
construction~\cite{DK06}.

During the execution of the computation, the adapted measurement angle
$\phi'_{x,y}$ is computed from $\phi_{x,y}$ and the previous
measurement outcomes in the following way: let $s^X_{x,y} =
\oplus_{i\in D_{x,y}}{s_i}$ be the parity of all measurement outcomes
for qubits in $D_{x,y}$ and similarly, $s^Z_{x,y} = \oplus_{i\in
  D'_{x,y}}{s_i}$ be the parity of all measurement outcomes for qubits
in~$D'_{x,y}$ (the index $i$ here is a two coordinate index, an
element of $[n] \times \{0,\dotsc,m\}$). 
Then, \begin{equation} \label{PhiEq} \phi'_{x,y} =
  (-1)^{s^X_{x,y}} \phi_{x,y} + s^Z_{x,y} \pi.\end{equation}

This will be used in a protocol, where the first column of the
brickwork state is a one-time pad encryption of the input.\footnote{In
  UBQC with a quantum input, the input is initially encoded with a
  variant of the quantum one-time pad by Alice, to preserve her
  privacy.  The operators implementing the one-time pad that Alice
  applies to the input may include an arbitrary rotation within the
  $XY$ plane of the Bloch sphere (a $Z_{\theta}$ rotation), and a
  Pauli-X operator.  Because of the commutation relation
  $(X\otimes\id) \ctrlZ = \ctrlZ (X\otimes Z)$ between the Pauli-X
  operator and the controlled $Z$ entangling operation, this component
  of the one-time pad must be accounted for in the measurement angles
  for the neighbors of the input layer, as in
  \eqnref{PhiEqOTP}.} The measurement angles of the first two columns
then have to be updated to compensate for (bit) flips $i_x$ performed
by the encryption, namely
\begin{equation} \label{PhiEqOTP}
\phi'_{x,0}=(-1)^{i_{x}}\phi_{x,0} \quad \text{ and } \quad \phi'_{x,1}=\phi_{x,1} + {i_{x}} \pi.
\end{equation}

\begin{algorithm}[htbp]
\caption[UBQC protocol]{Universal Blind Quantum Computation}

\begin{flushleft}
\textbf{Alice's input:}
\vspace{-7pt}
\end{flushleft}
\begin{itemize}
 \setlength{\itemsep}{-1pt}
 \item An $n-$qubit unitary map $U$, represented as a sequence of measurement angles  $\{\phi_{x,y}\}$ of a one-way quantum computation over a brickwork state of the size
 $n \times (m+1)$, along with the $X$ and $Z$ dependency sets 
 $D_{x,y}, D_{x,y}^{'}$, respectively.
\item An $n$-qubit input state $\rho_{in}$
\end{itemize}
\begin{flushleft}
\textbf{Alice's output (for an honest Bob):}
\vspace{-7pt}
\end{flushleft}
\begin{itemize}
 \setlength{\itemsep}{-1pt}
 \item The $n-$qubit quantum state $\rho_{out} = U \rho_{in} U^\dagger$
\end{itemize}
\begin{flushleft}
\textbf{The protocol}
\vspace{-7pt}
\end{flushleft}
\begin{enumerate}
 \setlength{\itemsep}{-1pt}
\item  \label{step:client-prep} \textbf{State preparation} 
\vspace{-7pt}
\begin{enumerate}
 \setlength{\itemsep}{-1pt}

\item \label{step:otp} For each $x \in [n]$, Alice applies
  $X^{i_{x}} Z_{\theta_{x,0}}$ to the $\ith{x}$ qubit of the input
  $\rho_{in}$, where the binary values $i_{x}$ and the angles
  $\theta_{x,0} \in \{k \pi / 4\}_{k=0}^{7}$ are chosen uniformly at
  random for each $x$. This is equivalent to encrypting it with a
  quantum one-time pad. The result is sent to Bob.
\item If $i_x = 1$, Alice updates the measurement angles $\phi_{x,0}$
  and $\phi_{x,1}$ to compensate for the introduced bit flip (see
  \eqnref{PhiEqOTP}).
\item \label{step:theta} For each column $y \in [m-1]$, and each row
  $x \in [n]$, Alice prepares the state $ \sket{+_{\theta_{x,y}}}:=
  \frac{1}{\sqrt{2}}(\sket{0} + e^{i\theta_{x,y}}\sket{1})$, where the
  defining angle $\theta_{x,y}\in \{k \pi / 4\}_{k=0}^{7}$ is chosen
  uniformly at random, and sends the qubits to Bob.
\item Bob creates $n$ qubits in the $\sket{+}$ state, which are used
  as the final output layer, and entangles the qubits received from
  Alice and this final layer by applying $\ctrlZ$ operators
  between the pairs of qubits specified by the pattern of the brickwork
  state $\mathcal{G}_{n \times (m+1)}$.
\end{enumerate}
%
\item  \label{step:computation-interaction}
 \textbf{Interaction and measurement}
 \vspace{-3pt}

 For $y = 0, \ldots , m-1$, repeat\\
 \hspace*{\parindent}\hspace*{\parindent}  For $x = 1, \ldots, n$, repeat
\begin{enumerate}
 \setlength{\itemsep}{-1pt}
\item Alice computes the updated measurement angle $\phi'_{x,y}$ (see
  \eqnref{PhiEq}), to take previous measurement outcomes received from
  Bob into account.
\item  \label{step:delta}
Alice chooses a binary digit $r_{x,y} \in \{0,1\}$ uniformly at random, and computes $\delta_{x,y} = \phi'_{x,y}  + \theta_{x,y} + \pi r_{x,y}$.
\item Alice  transmits $\delta_{x,y}$ to Bob, who performs a measurement in the basis $\{ \sket{+_{\delta_{x,y}}},
\sket{-_{\delta_{x,y}}} \}$.
\item Bob transmits the result $s_{x,y} \in \{0,1\}$ to Alice.
\item If $r_{x,y} = 1$, Alice flips $s_{x,y}$;
otherwise she does nothing.
\end{enumerate}
 \setlength{\itemsep}{-1pt}
 \item  \label{OutputCorrection} \textbf{Output Correction}
 \vspace{-7pt}
\begin{enumerate}
 \setlength{\itemsep}{-1pt}
\item Bob sends to Alice all qubits in the last (output) layer.
\item Alice performs the final Pauli corrections $ \{Z^{s_{x,m}^{Z} }X^{s_{x,m}^{X}} \}_{x=1}^{n}$ on the received output qubits.
\end{enumerate}

\end{enumerate}

\label{prot:UBQC}
\end{algorithm}

%


Protocol \ref{prot:UBQC} implements a blind quantum computation for an
input $\psi_A = \rho_{in} \tensor \proj{U}$.\footnote{The particular
  variant of the UBQC protocol we present assumes a quantum input and
  a quantum output, however the protocol is easily modified to take
  classical inputs and/or produces classical outputs, see
  \cite{BFK09}. In the classical input case, the quantum input is
  simply not sent, and the preparation of the classical input is
  assumed to be encoded in the computation itself. For the classical
  output, the server would simply measure out the final column of
  qubits as well, which produces a one-time padded version of the
  computation result. The quantum input-output setting is more general
  than other variants, and the security of this variant implies the
  security of the classical input/output versions. Also, the quantum
  one-time pad of the input states used in this protocol could be
  replaced with a standard quantum one-time pad which uses only the
  local $X$ and $Z$ gates, instead of the $X$ and the parametrized
  $Z_\theta$ gate, as presented here. In this case Bob would teleport
  the input state onto the brickwork state built out of the
  pre-rotated $\sket{+_\theta}$ qubits, and the protocol would
  continue as we have presented (but taking into account the
  teleportation outcomes reported by Bob).
}  It
was shown in \cite{BFK09} that this protocol is correct, i.e., if both
Alice and Bob follow the steps of the protocol then the final output
state is $\rho_{out} = U \rho_{in} U^\dagger$.

\subsubsection{One-time pad proof sketch}
\label{sec:blind.bfk.otp}

The basic idea behind the construction of the simulator required for
the proof of composable security of the UBQC protocol can be used in
the case of a simpler protocol \--- the Quantum One-Time Pad
(QOTP). The QOTP ensures confidentiality, but not authenticity, of the
exchange of quantum messages over an untrusted quantum channel.

The ideal confidentiality resource $\aS$, which we wish to construct,
has three interfaces, $A$ (Alice, the sender), $B$ (Bob, the receiver)
and $E$ (Eve, the eavesdropper). Alice inputs a message
$\rho^{\inn}_A$, Eve only learns the message size \--- though for
simplicity, we assume that the message size is fixed, and do not model
it explicitly in the following \--- but can arbitrarily modify or
replace the message. Similarly to the blind DQC ideal resource
(\defref{def:dqc.b}), the eavesdropper's capacity to arbitrarily
manipulate the message is captured by allowing some arbitrary state
$\rho^{\inn}_E$ and a description of a map $\cE : \lo{AE} \to \lo{B}$ to be
input at the $E$\=/interface of the ideal resource, which then outputs
$\cE(\rho_{AE})$ at the $B$\=/interface. This is depicted in
\figref{fig:qotp.ideal} with Eve's functionalities grayed to signify
that they are only accessible to a cheating player.

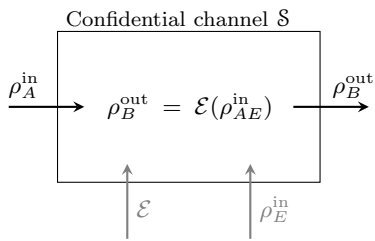
\begin{figure}[htb]
\begin{centering}

\begin{tikzpicture}\small

\def\t{2.368} 
\def\u{-1.75}
\def\v{.809}

\node[draw,text width=3.236cm,minimum height=2cm,text centered] (channel) at (0,0)
{$\rho^{\out}_B = \cE(\rho^{\inn}_{AE})$};
\node[yshift=-1.5,above right] at (channel.north west) {{\footnotesize
    Confidential channel} $\aS$};
\node[minimum width=1.118cm,minimum height=1.5cm] (innerl) at (-\v,0) {};
\node[minimum width=1.118cm,minimum height=1.5cm] (innerr) at (\v,0) {};
\node (alice) at (-\t,0) {};
\node (bob) at (\t,0) {};

\draw[sArrow] (alice.center) to node[pos=.2,auto] {$\rho^{\inn}_A$} (innerl);
\draw[sArrow] (innerr) to node[pos=.8,auto] {$\rho^{\out}_B$} (bob.center);
\node (eleft) at (-\v,\u) {};
\node (eright) at (\v,\u) {};
\draw[gArrow] (eleft.center) to node[auto,pos=.4,swap] {$\cE$} (innerl);
\draw[gArrow] (eright.center) to node[auto,pos=.4,swap] {$\rho^{\inn}_E$}  (innerr);

\end{tikzpicture}

\end{centering}
\caption[Quantum confidential channel]{\label{fig:qotp.ideal}A
  confidential channel. Alice and Bob have access to the left and
  right interface, respectively, and Eve accesses the lower
  interface. This channel guarantees that Eve does not learn Alice's
  input $\rho^{\inn}_A$, but allows her to modify what Bob
  receives. If Eve does not activate her cheating interface, the state
  $\rho^{\inn}_A$ is output at Bob's interface.}
\end{figure}

The resources $\aR$ available to the QOTP protocol $(\pi_A,\pi_B)$ are
a shared secret key and an insecure quantum channel, which simply
outputs at the $E$\=/interface anything which Alice inputs, and
forwards to the $B$\=/interface anything which Eve inputs. $\pi_A$
applies bit and phase flips (conditioned on the bits of the secret
key) to Alice's input and sends the result down the insecure channel,
and $\pi_B$ decrypts by applying the same flips to whatever it
receives. This is illustrated in \figref{fig:qotp.real}.

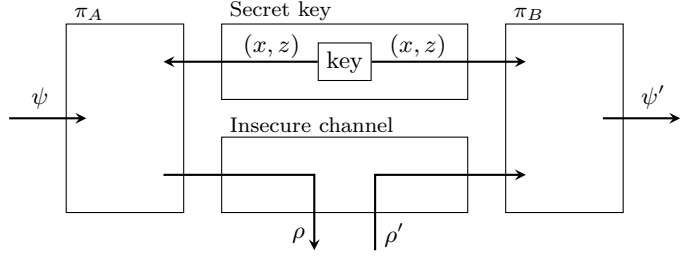
\begin{figure}[htb]
\begin{centering}

\begin{tikzpicture}\small

\def\t{4.413} 
\def\u{2.89} 
\def\v{.75}

\node[pnode] (a1) at (-\u,\v) {};
\node[pnode] (a2) at (-\u,0) {};
\node[pnode] (a3) at (-\u,-\v) {};
\node[protocol] (a) at (-\u,0) {};
\node[yshift=-2,above right] at (a.north west) {\footnotesize
  $\pi_A$};
\node (alice) at (-\t,0) {};

\node[pnode] (b1) at (\u,\v) {};
\node[pnode] (b2) at (\u,0) {};
\node[pnode] (b3) at (\u,-\v) {};
\node[protocol] (b) at (\u,0) {};
\node[yshift=-2,above right] at (b.north west) {\footnotesize $\pi_B$};
\node (bob) at (\t,0) {};

\node[thinResource] (keyBox) at (0,\v) {};
\node[draw] (key) at (0,\v) {key};
\node[yshift=-2,above right] at (keyBox.north west) {\footnotesize Secret key};
\node[thinResource] (channel) at (0,-\v) {};
\node[yshift=-1.5,above right] at (channel.north west) {\footnotesize
  Insecure channel};
\node (eveleft) at (-.4,-1.75) {};
\node (everight) at (.4,-1.75) {};
\node (ajunc) at (eveleft |- a3) {};
\node (bjunc) at (everight |- b3) {};

\draw[sArrow] (key) to node[auto,swap,pos=.3,yshift=-2] {$(x,z)$} (a1);
\draw[sArrow] (key) to node[auto,pos=.3,yshift=-2] {$(x,z)$} (b1);

\draw[sArrow] (alice.center) to  node[auto,pos=.4] {$\psi$} (a2);
\draw[sArrow] (b2) to node[auto,pos=.65] {$\psi'$} (bob.center);

\draw[sArrow] (a3) to (ajunc.center)
to node[pos=.8,auto,swap] {$\rho$} (eveleft.center);
\draw[sArrow] (everight.center) to node[pos=.2,auto,swap] {$\rho'$}
(bjunc.center) to (b3);

\end{tikzpicture}

\end{centering}
\caption[Quantum one-time pad]{\label{fig:qotp.real}The concrete
  setting of the QOTP, with Alice accessing the left interface, Bob
  the right one and Eve the lower interface. The QOTP encrypts a
  message $\psi$ by applying bit and phase flips, $\rho \coloneqq Z^z
  X^x \psi X^x Z^z$, and decrypts by applying the reverse operation,
  $\psi' \coloneqq X^x Z^z \rho' Z^z X^x$.}
\end{figure}

To prove that this protocol constructs the ideal confidentiality
resource, we need to find a simulator $\sigma_E$ that, when plugged
into the $E$\=/interface of the ideal resource, emulates the
communication on the insecure quantum channel and finds the
appropriate inputs $\rho^{\inn}_E$ and $\cE$ that correspond to Eve's
tampering, so that ideal and concrete cases are indistinguishable. In
other words, we need to find a $\sigma_E$ such that
\begin{equation} \label{eq:QOTPconfidentiality}
  \pi_A\aR\pi_B = \aS\sigma_E.
\end{equation}

In the concrete setting, the distinguisher accessing $\pi_A\aR\pi_B$
can choose an arbitrary input $\rho^{\inn}_{AR}$, apply an arbitrary
map $\cD$ to the state on the quantum channel (output at the
$E$\=/interface) and its own system $R$, and put the result back on
the quantum channel. After decryption by $\pi_B$, it ends up with the
final state $\rho^{\out}_{BR}$. We depict this for one-qubit messages
in \figref{QOTP:fig}, by rearranging \figref{fig:qotp.real} as a
circuit with the addition of the purifying system $R$ and map $\cD$.

  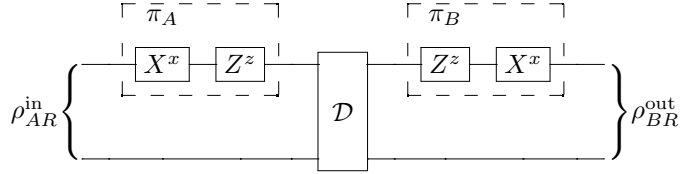
\begin{figure}[htb]
  \[\hspace{0.3cm}
  \Qcircuit @C=1em @R=0.5em @!R {
 &&&&&&& \pi_A&&&&&\pi_B&&&&&&&&&\\
 &&&&&   &\qw &\gate{X^x} & \gate{Z^z} & \qw & \multigate{2}{\cD} & \qw &\gate{Z^{z}} & \gate{X^{x}}& \qw &\qw &\\
  &&&&\rho_{AR}^{\inn}   \vast\lbrace\ \ &&     &  & &                 &  &  && & &&\ \ \ \vast\rbrace \rho_{BR}^{\out}  \\
  &&&&& & \qw & \qw &\qw&\qw& \ghost{D} &\qw&\qw&\qw&\qw&\qw\gategroup{1}{8}{2}{9}{1em}{--}\gategroup{1}{13}{2}{14}{1em}{--}&
  }
  \] 
  \caption[Quantum one-time pad as circuit]{\label{QOTP:fig}
    Interaction of the distinguisher and the QOTP.}
  \end{figure}

  In the ideal setting, the simulator $\sigma_E$ needs to simulate the
  quantum channel and provide the ideal resource $\aS$ with
  information allowing it to generate the same output
  $\rho^{\out}_{BR}$ as in the concrete case. It does this by
  outputting half an EPR pair (for every qubit of the message) at its
  outer interface, and transmitting the other half along with any
  state it received at its outer interface to the ideal resource. It
  also provides the ideal resource with the ``instructions'' $\cE$ to
  gate teleport the real input through the map $\cD$ of the
  distinguisher, i.e., it teleports the input using the EPR half,
  registers the possible bit and phase flips, and outputs the second
  state received after having corrected the bit and phase flips from
  the teleportation. Plugging this simulator into the $E$\=/interface
  of \figref{fig:qotp.ideal} along with the distinguisher's input
  $\rho^{\inn}_{AE}$ and map $\cD$, and rewriting it as a circuit for
  one-qubit messages results in \figref{QOTP.ideal:fig}.

 \begin{figure}[htb]
  \begin{centering}
  \includegraphics[clip,scale=1.05]{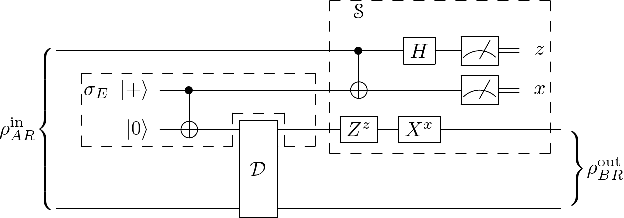} 

  \end{centering}
  \caption[Confidential channel and simulator as
  circuit]{\label{QOTP.ideal:fig}Interaction of the ideal
    confidentiality resource $\aS$ and the simulator $\sigma_E$ with
    the distinguisher. $\aS$ does not leak any information to the
    adversary, it receives inputs from Alice ($\rho_A^{\inn}$) and the
    simulator, and transmits some state to Bob. $\sigma_E$ \--- which
    does give information to the adversary \--- has no access to the
    confidential message $\rho_A^{\inn}$.}
\end{figure}

We now show that the circuits from \figref{QOTP:fig} and
\figref{QOTP.ideal:fig} are indistinguishable, hence
\eqnref{eq:QOTPconfidentiality} holds. The argument generalizes
straightforwardly to multiple qubit messages. We first rearrange
\figref{QOTP.ideal:fig} by grouping the state preparation (performed
by $\sigma_E$) and the actual teleportation (performed by $\aS$). This
results in \figref{QOTP.intermediate:fig}.

 \begin{figure}[htb]
  \[\hspace{-0.2cm}
  \Qcircuit @C=1em @R=0.5em @!R {
  &&&&&&&&&&&&&&&&\\
  &&\qw &\qw&\qw&\qw&\qw& \qw    & \ctrl{1}& \gate{H} & \meter &\cw & \hspace{-0.cm}z\\
  &&&&&&\lstick{\ \ \, \ket{+}} &\ctrl{1}&   \targ  & \qw & \meter&\cw &\hspace{-0.cm}x\\
 & \rho_{AR}^{\inn} \Vast\lbrace\hspace{10mm}&&&&&\lstick{\ket{0}}  &\targ   &\qw &\qw & \qw&\qw& \qw& \multigate{2}{\cD}& \qw &\gate{Z^{z}} & \gate{X^{x}} &\qw \\
 & &&&&&&& &&&&&&&&&&\hspace{5mm}\vast\rbrace\rho^{\out}_{BR}\\
   & &\qw &\qw&\qw&\qw&\qw   & \qw &\qw &\qw &\qw & \qw &\qw& \ghost{B}&\qw&\qw &\qw &\qw  \gategroup{2}{6}{4}{12}{3em}{--}
  }
  \]
  \caption[Quantum one-time pad and confidential channel
  hybrid]{\label{QOTP.intermediate:fig} Reformulation of
    \figref{QOTP.ideal:fig} by grouping the simulator and the
    teleportation step of the ideal confidentiality resource. The
    circuit in the dashed box simply encrypts the input with a random
    bit and phase flip, and therefore corresponds to $\pi_A$.}
  \end{figure}
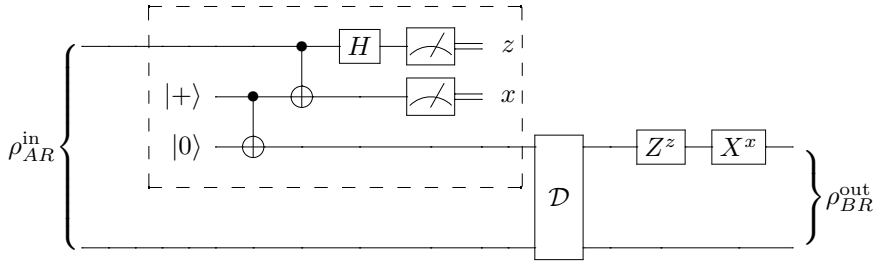
  
The circuit in the dashed box of \figref{QOTP.intermediate:fig}
teleports the input from the first wire to the third wire (without
correcting the random flips). This is equivalent to simply
performing a random bit and phase flip on the input, which is
exactly what is done by the QOTP in \figref{QOTP:fig}.

\subsubsection{Security}
\label{sec:blind.bfk.security}

In this section we prove that the UBQC protocol
(Protocol~\ref{prot:UBQC}) provides perfect blindness, i.e., we find a
simulator $\sigma_B$ such that the two interactive boxes in
\figref{fig:dqc.sec} are indistinguishable. Similarly to the one-time
pad proof sketch from \secref{sec:blind.bfk.otp}, we construct a
simulator which sends only EPR pair halves and random strings, then
transmits the other halves and the transcript to the ideal blind DQC
resource. Whenever a one-time padded quantum state should have been
sent, the ideal resource teleports it using the EPR half, and uses the
bit and phase flips of the teleportation as one-time pad key. And
whenever a random string $r$ was sent instead of some one-time padded
string $s$, the ideal resources sets $r \oplus s$ as the random key
used to encrypt and send $s$.

To prove that the real and ideal settings are identical, we replace
steps of the protocol by equivalent steps, until we end up with the
desired simulator and ideal resource.

Protocol~\ref{prot:UBQC} does not explicitly model the information
that is intentionally allowed to leak. This information consists in
the size of the brickwork state (which leaks upper bounds on the input
state size and computation size), and whether the last column
of the brickwork state should be measured, i.e., whether the output of
the protocol is classical or quantum. It is simply assumed that this
information is known by the server (Bob), otherwise it could not
perform the desired computation. For simplicity we also avoid modeling
this information in the following. The protocol and proof can however
be trivially changed to include it.

\begin{thm}
  \label{thm:bfk}
  The DQC protocol described in Protocol \ref{prot:UBQC} provides perfect blindness.
\end{thm}

\begin{proof}
  To prove that $\pi_A\aR = \aS^{\blind}\sigma_B$, we successively
  modify the protocol $\pi_A$, replacing some steps with equivalent
  steps that implement the same map, resulting in several intermediary
  protocols, until we achieve a version which corresponds to
  $\aS^{\blind}\sigma_B$.

\begin{algorithm}[htb]

\caption[Equivalent UBQC protocol I]{UBQC, equivalent protocol for Alice, first version}

\begin{flushleft}
\textbf{The protocol}
\vspace{-7pt}
\end{flushleft}
\begin{enumerate}
 \setlength{\itemsep}{-1pt}
\item \textbf{State preparation} 
\vspace{-7pt}
\begin{enumerate}
 \setlength{\itemsep}{-1pt}

\item \label{step:otp2} For each $x \in [n]$, Alice prepares
  an EPR pair $(\sket{00}+\sket{11})/\sqrt{2}$ and sends half to
  Bob. She picks an angle $\theta'_{x,0} \in \{k \pi / 4\}_{k=0}^{7}$ uniformly at
  random, and applies $Z_{\theta'_{x,0}}$ to the $\ith{x}$ qubit of the input
  $\rho_{in}$. She then teleports the resulting qubit using her half
  of the EPR pair, and registers the values of the bit and phase flips
  resulting from the teleportation in $i_{x}$ and $r_{x,0}$,
  respectively.

\item If $i_x = 1$, Alice updates the measurement angles $\phi_{x,0}$
  and $\phi_{x,1}$ (see \eqnref{PhiEqOTP}).

\item \label{step:theta2} For each column $y \in [m-1]$, and
  each row $x \in [n]$, Alice prepares an EPR pair
  $(\sket{00}+\sket{11})/\sqrt{2}$ and sends half to Bob. She then
  picks an angle $\theta'_{x,y}\in \{k \pi / 4\}_{k=0}^{7}$ uniformly
  at random, performs a $Z_{\theta'_{x,y}}$ rotation followed by a
  Hadamard $H$ on her half of the pair, and measures it in the
  computational basis. She stores the result in $r_{x,y}$.
\end{enumerate}

\item \textbf{Interaction and measurement}
 \vspace{-3pt}

 For $y = 0, \ldots , m-1$, repeat\\
 \hspace*{\parindent}\hspace*{\parindent}  For $x = 1, \ldots, n$, repeat
\begin{enumerate}
 \setlength{\itemsep}{-1pt}
\item Alice computes the updated measurement angle $\phi'_{x,y}$ (see
  \eqnref{PhiEq}).
\item \label{step:delta2} Alice computes $\delta_{x,y} = \phi'_{x,y} +
  \theta'_{x,y}$ and transmits this to Bob.
\item Alice receives a bit $s_{x,y} \in \{0,1\}$ from Bob.
\item If $r_{x,y} = 1$, Alice flips $s_{x,y}$;
otherwise she does nothing.
\end{enumerate}
 \setlength{\itemsep}{-1pt}
 \item \textbf{Output Correction}
 \vspace{-7pt}
\begin{enumerate}
 \setlength{\itemsep}{-1pt}
\item Alice receives $n$ qubits from Bob, and performs the final Pauli
  corrections $\{Z^{s_{x,m}^{Z} }X^{s_{x,m}^{X}} \}_{x=1}^{n}$ on
  these qubits.
\end{enumerate}

\end{enumerate}

\label{prot:UBQC2}
\end{algorithm}

The first intermediary protocol is given by
Protocol~\ref{prot:UBQC2}. Compare Step~\ref{step:otp} of
Protocol~\ref{prot:UBQC} and Step~\ref{step:otp2} of
Protocol~\ref{prot:UBQC2}. In the former, Alice picks random values
$\theta_{x,0}$ and $i_x$ and performs corresponding phase and bit
rotations on the $\ith{x}$ input qubit. In the latter, she performs a
random $\theta'_{x,0}$ phase rotation, and teleports the resulting
state. For teleportation outcomes $i_x$ and $r_{x,0}$, and setting
$\theta_{x,0} \coloneqq \theta'_{x,0} + \pi r_{x,0}$, Bob holds
exactly the same state. Since the different values of $i_x$ and
$\theta_{x,0}$ occur with the same (uniform) probabilities in both
protocols, these implement identical maps.

Likewise, compare Step~\ref{step:theta} of Protocol~\ref{prot:UBQC}
and Step~\ref{step:theta2} of Protocol~\ref{prot:UBQC2}. In the former
Alice sends a state $\sket{+_{\theta_{x,y}}}$ to Bob; in the latter
Bob ends up holding the state $\sket{+_{\theta'_{x,y}+\pi
    r_{x,y}}}$. If Alice sets $\theta_{x,0} \coloneqq \theta'_{x,0} +
\pi r_{x,0}$ in her internal memory, all states of the systems are
identical for both protocols.

Finally, the only other difference between these protocols is in
Steps~\ref{step:delta} and \ref{step:delta2} of the two protocols,
respectively.  In the former, Alice sends Bob the angle $\phi'_{x,0} +
\theta_{x,0} + \pi r_{x,0}$, for some randomly picked bit $r_{x,0}$;
in the latter, she sends $\phi'_{x,0} + \theta'_{x,0}$. But as we've
already established, these two angles are identical, and occur with
the same (uniform) probabilities.

\begin{algorithm}[htb]

  \caption[Equivalent UBQC protocol II]{UBQC, equivalent protocol for
    Alice, second version}

\begin{flushleft}
\textbf{The protocol}
\vspace{-7pt}
\end{flushleft}
\begin{enumerate}
 \setlength{\itemsep}{-1pt}
\item \textbf{State preparation} 
\vspace{-7pt}
\begin{enumerate}
 \setlength{\itemsep}{-1pt}

\item \label{step:otp3} For each $x \in [n]$, Alice prepares
  an EPR pair $(\sket{00}+\sket{11})/\sqrt{2}$ and sends half to
  Bob. She performs the first measurement of a teleportation that
  determines the bit flip, i.e., for each $x$ she performs a CNOT on
  the corresponding EPR half using the input qubit as control, and measures
  the EPR half in the computational basis. She records the outcome in
  $i_{x}$.
\item If $i_x = 1$, Alice updates the measurement angles $\phi_{x,0}$
  and $\phi_{x,1}$ (see \eqnref{PhiEqOTP}).
\item \label{step:theta31} For each column $y \in [m-1]$, and each
  row $x \in [n]$, Alice prepares an EPR pair
  $(\sket{00}+\sket{11})/\sqrt{2}$ and sends half to Bob.
\end{enumerate}

\item \textbf{Interaction and measurement}
 \vspace{-3pt}

 For $y = 0, \ldots , m-1$, repeat\\
 \hspace*{\parindent}\hspace*{\parindent}  For $x = 1, \ldots, n$, repeat
\begin{enumerate}
 \setlength{\itemsep}{-1pt}
\item Alice computes the updated measurement
  angle $\phi'_{x,y}$ (see \eqnref{PhiEq}).
\item \label{step:delta3}
Alice picks an angle $\delta_{x,y}\in \{k \pi /
  4\}_{k=0}^{7}$ uniformly at random, and sends it to Bob.
\item Alice receives a bit $s_{x,y} \in \{0,1\}$ from Bob.
\item \label{step:theta32} Alice computes $\theta'_{x,y} = \delta_{x,y} - \phi'_{x,y}$. She
  then applies $Z_{\theta'_{x,y}}$, followed by a Hadamard $H$ and a
  measurement in the computational basis to the $\ith{x}$ qubit of
  the input $\rho_{in}$ if $y = 0$, and to the corresponding EPR half
  if $y > 0$. She stores the result in $r_{x,y}$.
\item If $r_{x,y} = 1$, Alice flips $s_{x,y}$;
otherwise she does nothing.
\end{enumerate}
 \setlength{\itemsep}{-1pt}
 \item \textbf{Output Correction}
 \vspace{-7pt}
\begin{enumerate}
 \setlength{\itemsep}{-1pt}
\item Alice receives $n$ qubits from Bob, and performs the final Pauli
  corrections $ \{Z^{s_{x,m}^{Z} }X^{s_{x,m}^{X}} \}_{x=1}^{n}$ on
  these qubits.
\end{enumerate}

\end{enumerate}

\label{prot:UBQC3}
\end{algorithm}

Now, compare Protocol~\ref{prot:UBQC2} and
Protocol~\ref{prot:UBQC3}. The main difference is between
Step~\ref{step:delta2} of Protocol~\ref{prot:UBQC2} and
Step~\ref{step:delta3} of Protocol~\ref{prot:UBQC3}. In the former,
Alice had picked $\theta'_{x,y}$ uniformly at random, and sends Bob
$\delta_{x,y}$, a one-time padded version of $\phi'_{x,y}$ with
$\theta'_{x,y}$ as the key; hence $\delta_{x,y}$ is uniformly
distributed. In the latter protocol, Alice instead picks
$\delta_{x,y}$ uniformly at random (in Step~\ref{step:delta3}), then
computes $\theta'_{x,y} \coloneqq \delta_{x,y} - \phi'_{x,y}$ (in
Step~\ref{step:theta32}) to get the value of the uniform key used to
encrypt $\phi'_{x,y}$.

In Protocol~\ref{prot:UBQC2}, Alice used the value of $\theta'_{x,y}$
in Steps~\ref{step:otp2} and \ref{step:theta2}. Since $\theta'_{x,y}$
is not available at those stages of Protocol~\ref{prot:UBQC3}, the
corresponding steps are delayed until this value is available. Hence
Step~\ref{step:otp3} of Protocol~\ref{prot:UBQC3} only consists in
performing the first part of the teleportation (which commutes with
the $Z_{\theta'_{x,0}}$ rotation) and in Step~\ref{step:theta31} Alice
only sends half an EPR pair. In Step~\ref{step:theta32}, after
computing $\theta'_{x,y}$, Alice completes those two steps by
performing the missing operations.

\begin{algorithm}[htbp]

\caption[UBQC resource and simulator]{UBQC, simulator and ideal resource}

\begin{flushleft}
\textbf{The simulator}
\vspace{-7pt}
\end{flushleft}
\begin{enumerate}[label=\arabic*., ref=\arabic*]
 \setlength{\itemsep}{-1pt}

\item For each column $y \in \{0, \ldots ,
  m-1\}$, and each row $x \in [n]$, the simulator prepares an EPR pair
  $(\sket{00}+\sket{11})/\sqrt{2}$ and outputs half at its outer
  interface.

\item For each column $y \in \{0, \ldots ,
  m-1\}$, and each row $x \in [n]$, the simulator picks an angle
  $\delta_{x,y}\in \{k \pi / 4\}_{k=0}^{7}$ uniformly at random, and
  outputs it at its outer interface. It receives some response
  $s_{x,y} \in \{0,1\}$.

\item The simulator receives $n$ qubits, which correspond to the last
  (output) layer.

\item \label{step:IR-message} The simulator transmits all EPR pair half, all angles
  $\delta_{x,y}$, bits $s_{x,y}$ and output qubits to the ideal blind delegated
  quantum computation resource, along with instructions to perform the
  operations described hereafter.
\end{enumerate}

\begin{flushleft}
\textbf{The ideal blind DQC resource}
\vspace{-7pt}
\end{flushleft}
\begin{enumerate}[label=\arabic*., ref=\arabic*]
 \setlength{\itemsep}{-1pt}

\item The blind DQC resource receives the input $\rho_{in}$ and
  a description  of the computation given by angles $\phi_{x,y}$ at its
  $A$\=/interface, and all the information described in
  Step~\ref{step:IR-message} above at its $B$\=/interface.
\item For each $x \in [n]$, it performs the first
  measurement of a teleportation of the input, i.e.,
  for each $x$ it performs a CNOT on the corresponding EPR half using
  the input qubit as control, and measures the EPR half in the
  computational basis. It records the outcome in $i_{x}$.
\item If $i_x = 1$, it updates the measurement angles $\phi_{x,0}$
  and $\phi_{x,1}$ (see \eqnref{PhiEqOTP}).

\item  For $y = 0, \ldots , m-1$, repeat\\
 \hspace*{\parindent}\hspace*{\parindent}  For $x = 1, \ldots, n$, repeat
\begin{enumerate}
\item It computes the updated measurement
  angle $\phi'_{x,y}$ (see \eqnref{PhiEq}).
\item It computes $\theta'_{x,y} = \delta_{x,y} - \phi'_{x,y}$. It
  then applies $Z_{\theta'_{x,y}}$, followed by a Hadamard $H$ and a
  measurement in the computational basis to the $\ith{x}$ qubit of
  the input $\rho_{in}$ if $y = 0$, and to the corresponding EPR half
  if $y > 0$. It stores the result in $r_{x,y}$.
\item If $r_{x,y} = 1$, it flips $s_{x,y}$;
otherwise it does nothing.
\end{enumerate}

\item The ideal blind DQC resource performs the final Pauli
  corrections $ \{Z^{s_{x,m}^{Z} }X^{s_{x,m}^{X}} \}_{x=1}^{n}$ on the
  received output qubits, and outputs the result at its $A$\=/interface.
\end{enumerate}

\label{prot:UBQC4}
\end{algorithm}

Protocol~\ref{prot:UBQC4} consists in exactly the same steps as
Protocol~\ref{prot:UBQC3}, but their order has been rearranged, and
the different parts have been renamed ``simulator'' and ``ideal
resource''. The ideal blind DQC resource constructed meets the
requirements of \defref{def:dqc.b}, we have $\pi_A\aR = \aS^{\blind}\sigma_B$
and conclude the proof.
\end{proof}

\subsection{One-way communication}
\label{sec:blind.oneway}

If a protocol only requires one-way communication from Bob to Alice,
the protocol model described in \secref{sec:dqc.model.concrete} can be
simplified: it only consists in two operations. Bob generates a state
$\tau$, which he sends to Alice on the channel $C$. She then applies
some operation $\cE : \lo{AC} \to \lo{A}$ to her input and $\tau$, and
outputs the contents of her system $A$.

\begin{thm}
  \label{thm:oneway}
  Any DQC protocol $\pi$ with one-way communication from Bob
  to Alice provides perfect blindness.
\end{thm}

\begin{proof}
  The simulator $\sigma_B$ works as follows. It receives some state
  $\psi_C$ from the distinguisher, and provides it to the ideal
  resource $\aS^{\blind}$ along with a description of the map $\cE$
  that is used by $\pi_A$. Alice's output is thus $\cE(\psi_{AC})$,
  and we immediately have $d(\pi_AR,S\sigma_B) = 0$.
\end{proof}

This proof does not mention the permitted leaks at the
$B$\=/interface. This is because protocols with one-way communication
make the (implicit) assumption that this information is known to the
server. Alternatively, one could include a single message from Alice
to Bob containing this information, and adapt the proof above
accordingly.

\appendix
\appendixpage

\section{Distance measures for subnormalized states}
\label{app:dist}

In \secref{sec:systems.dist} we introduced the trace distance
$D(\rho,\sigma)$ between two quantum states. Another widely used
measure is the fidelity, defined as
\[ F(\rho,\sigma) \coloneqq \trace{\sqrt{\rho^{1/2}\sigma\rho^{1/2}}}.\]

When dealing with subnormalized states, we need to generalize these
measures to retain their properties. The following distance notions
are treated in detail in \cite{TCR10}, and we refer to that work for
more information.

For any two subnormalized states $\rho,\sigma \in \sno{}$, we define
the generalized trace distance as
\[\bar{D}(\rho,\sigma) \coloneqq D(\rho,\sigma) +
\frac{1}{2}|\tr \rho - \tr \sigma|,\] and the generalized fidelity as
\[\bar{F}(\rho,\sigma) \coloneqq F(\rho,\sigma) +
\sqrt{(1-\tr \rho)(1 - \tr \sigma)}.\]

The (generalized) fidelity has a useful property, known as Uhlmann's
theorem (see \cite{NC00} or \lemref{lem:uhlmann} here below), which
states that for any two states $\rho,\sigma$, there exist
purifications of these states which have the same fidelity. We define
a metric, the \emph{purified distance}, based on the fidelity, so as
to retain this property:
 \[P(\rho,\sigma) \coloneqq \sqrt{1-\bar{F}^2(\rho,\sigma)}.\]

This metric coincides with the generalized distance for pure states,
and is larger otherwise.
\begin{lem}[See \protect{\cite[Lemma 6]{TCR10}}]
\label{lem:purified}
Let $\rho,\sigma \in \sno{}$. Then \[\bar{D}(\rho,\sigma) \leq
P(\rho,\sigma) \leq \sqrt{2\bar{D}(\rho,\sigma)}.\]\end{lem}

Uhlmann's theorem restated for the purified distance is as follows.
\begin{lem}[See \protect{\cite[Lemma 8]{TCR10}}]
\label{lem:uhlmann}
Let $\rho,\sigma \in \sno{A}$ and $\varphi \in \sno{AR}$ be a
purification of $\rho$. Then there exists a purification $\psi \in
\sno{AR}$ of $\sigma$ such that $P(\rho,\sigma) =
P(\varphi,\psi)$.\end{lem}

\section{Correctness}
\label{app:cor}

Intuitively, a protocol is correct if, when Bob behaves honestly,
Alice ends up with the correct output. This must also hold with
respect to a purification of the input.

\begin{deff} \label{def:cor.sa} A DQC protocol provides
  $\eps$\=/local\-/correctness, if, when both parties behave honestly,
  for all initial states $\psi_{AR}$, the map implemented by the
  protocol on Alice's input, $\cP_A : \lo{A} \to \lo{A}$ is 
  \begin{equation}
    \label{eq:cor.sa}
    \cP_A \close{\eps} \cU.\end{equation}
\end{deff}

It is straightforward, that this is equivalent to the composable
notion defined in \eqnsref{eq:dqc.b} and \eqref{eq:dqc.v} in
\secref{sec:dqc.security}.

\begin{lem}
\label{lem:cor}
A DQC protocol which provides $\eps$\=/local\-/correctness is also
$\eps$\=/correct.
\end{lem}

\begin{proof}
  The resources $\pi_A\aR\pi_B$ and $\aS\bot_B$ have only one input
  and output, both on the $A$\=/interface, they are therefore maps
  $\lo{A} \to \lo{A}$. In fact, $\pi_A\aR\pi_B = \cP_A$ and $\aS\bot_B
  = \cU$. So from \defref{def:cor.sa}, $\pi_A\aR\pi_B \close{\eps}
  \aS\bot_B$.
\end{proof}

\section{Applying the reduction}
\label{app:app}

The definitions of local\-/blindness and local\-/verifiability used in
this work are equivalent to those used to prove local\-/security for
most protocols in the literature, e.g., by Fitzsimons and
Kashefi~\cite{FK12} and Morimae~\cite{Mor14}. To prove that such
protocols are secure, it remains to show that they satisfy the
stronger definition of \emph{independent} local\-/verifiability
introduced in this work. We sketch in this section that this is the
case for \cite{FK12} and \cite{Mor14}, and leave it open to prove this
formally.

\subsection{DQC protocol of Fitzsimons and Kashefi}
\label{app:app.FK}

Fitzsimons and Kashefi~\cite{FK12} extend the DQC protocol of
\cite{BFK09} to include a new approach which allows for verifiability
as well. They do this by suggesting a novel resource\-/state for
measurement\-/based quantum computing, the geometry of which allows
the random positioning of \emph{trap qubits} (the number of which can
be a fraction of the overall computation size). To achieve this, Alice
is additionally empowered to produce the Z observable eigenstates
$\ket{0}, \ket{1}$ along with the 8 symmetric states from the XY plane
of the Bloch sphere. They prove that if the measurement results of
these trap qubits are not what the client Alice expects, she knows
that the server is cheating, and if no traps are triggered, Alice can
be sure (up to some error $\eps$) that the server is running the
correct protocol.


\begin{lem}
\label{lem:FK}
If the protocol of \cite{FK12} is run with parameters such that it has
error $\eps$, then it is
$4\sqrt{2}\eps^{1/4}N^2$\=/blind\-/verifiable, where $N$ is the
dimension of the subsystem of Alice's which is quantum.
\end{lem}

\begin{proof}[Proof sketch.]
The protocol of \cite{FK12} is an extension of the UBQC protocol of
\cite{BFK09} analyzed in \secref{sec:blind.bfk}, and also provides
perfect blindness.

The verifiability definition used in \cite{FK12} is expressed
differently from that of \defref{def:sa.v}. For a pure input
$\sket{\psi_{AR}}$, the correct output is $\sket{\cU\psi_{AR}}
\coloneqq \cU \tensor \id_R \sket{\psi_{AR}}$. The projector
\[\Pi \coloneqq \id_{AR} - \proj{\cU\psi_{AR}} - \proj{\err} \tensor \id_{R}\] defines
the space where an erroneous output is accepted, and the verifiability
criterion of \cite{FK12} can be reduced to
\begin{equation} \label{eq:FKverif} \trace{\Pi\rho_{AR}} \leq \eps,
\end{equation}
where $\rho_{AR}$ is the state of Alice and the reference system at
the end of the protocol. Note that the output can always be written as
a linear combination of the error flag and some accepted output,
\[\rho_{AR} = p \sigma_{AR} + (1-p) \proj{\err} \tensor \psi_R.\]
Plugging this in the two definitions of local\-/verifiability we find
that \defref{def:sa.v} is equivalent to requiring
$pD(\sigma_{AR},\sket{\cU\psi_{AR}}) \leq \eps$ and
\eqnref{eq:FKverif} is equivalent to having
$p(1-F^2(\sigma_{AR},\sket{\cU\psi_{AR}}) \leq \eps$, where
$D(\cdot,\cdot)$ is the trace distance (\secref{sec:systems.dist}) and
$F(\cdot,\cdot)$ is the fidelity (\appendixref{app:dist}). Using
standard bounds between the trace distance and fidelity, we find that
any protocol which respects \eqnref{eq:FKverif} for all pure $AR$
inputs provides $\sqrt{\eps}$\=/local\-/verifiability.

To prove that the protocol satisfies perfectly independent
$\sqrt{\eps}$\=/local\-/verifiability, consider the proof technique
for the security of the UBQC protocol~\cite{BFK09} analyzed in
\secref{sec:blind.bfk}. There, we showed that instead of running the
correct protocol with Bob, Alice could equivalently run it using EPR
pairs instead of her quantum input. And once the interaction with Bob
is over, she finishes the computation locally by gate teleporting her
input through Bob's operations and obtains the same final
output. Since the protocol of \cite{FK12} is an extension of the UBQC
protocol of \cite{BFK09}, the same technique can be applied. However
instead of gate teleporting the input, we are interested here in
measuring the trap qubits (and ignore the other EPR pairs that could
be used for the gate teleportation and computation of the final
output). By doing this, Alice can determine if Bob is cheating,
without needing to have any input, and the verification mechanism is
thus clearly independent of the input. To formally prove that it
provides perfect independence, we need to find alternative maps that
Bob can apply, and which result in him holding $\rho_{\bar{A}B}$, the
joint system of Alice's decision to accept or reject his input,
$\bar{A}$, and his side information, $B$ (see
\remref{rem:sa.ind}). This can be done by disregarding the
communication with Alice, and running this alternative protocol with
EPR pairs on his own.

Putting this together with the fact that \cite{FK12} satisfies the
local\-/correctness condition and \corref{cor:sa.comp} concludes this
proof.
\end{proof}

\subsection{DQC protocol of Morimae}
\label{app:app.Mor}

Morimae~\cite{Mor14} generalizes the protocol of \cite{MF13} with
one-way communication from Bob to Alice (in which Alice measures the
individual qubits of the resource state sent to her by the sever Bob)
to include a notion of verifiability. In a first step, Alice runs the
same protocol as \cite{MF13}, but instead of computing the task
received as input, she runs an alternative computation that generates
in the last layer a new resource state with randomly positioned trap
qubits. In a second step, Alice measures the individual qubits of this
new resource, but this time with the goal of running the computation
provided as input. If no traps are triggered, she can be sure (up to
some error $\eps$), that the server is behaving honestly and her
outcome is correct.

Morimae discusses the local\-/blindness and local\-/verifiability of
this protocol.  Given these two properties, we only need to show that
this protocol is independent local-verifiable for our
\corref{cor:sa.comp} to be applicable. The argument is similar to the
proof sketch of \lemref{lem:FK}: Alice does not need to know the input
to measure the trap qubits and decide if Bob is cheating. Thus, Bob
could run the protocol on his own \--- without knowing Alice's input
and choosing himself the position of the trap qubits \--- measuring
only the trap qubits in the last layer, not those used for
computation. At the end of which, he would hold exactly the same bit
as Alice that decides if the output is accepted or rejected.

\section*{Acknowledgments}
\addcontentsline{toc}{section}{Acknowledgments}

This material is based on research supported in part by the Singapore
National Research Foundation under NRF Award No.~NRF-NRFF2013-01. VD
acknowledges the support of the EPSRC Doctoral Prize
Fellowship. Initial part of this work was performed while VD was at
Heriot-Watt University, Edinburgh, supported by EPSRC (grant
EP/E059600/1). CP and RR are supported by the Swiss National Science
Foundation (via grant No.~200020-135048 and the National Centre of
Competence in Research `Quantum Science and Technology') and the
European Research Council -- ERC (grant No.~258932).


\newcommand{\etalchar}[1]{$^{#1}$}
\providecommand{\bibhead}[1]{}
\expandafter\ifx\csname pdfbookmark\endcsname\relax%
  \providecommand{\tocrefpdfbookmark}{}
\else\providecommand{\tocrefpdfbookmark}{%
   \phantomsection%
   \addcontentsline{toc}{section}{\refname}}%
\fi

\tocrefpdfbookmark

\end{document}